\documentclass[12pt]{article}
\usepackage{amsthm}
\usepackage{amsmath,bm}
\usepackage{enumerate}
\usepackage{amssymb}
\usepackage{latexsym}
\usepackage{amsfonts}
\usepackage{subfigure}
\usepackage{color}
\usepackage{secdot}
\usepackage{mathrsfs}
\usepackage{epsfig}
\usepackage{natbib}
\usepackage{appendix}
\newtheorem{theorem}{Theorem}[section]
    \newtheorem{example}{Example}[section]
    \newtheorem{counterexample}{Counterexample}[section]
    \newtheorem{corollary}{Corollary}[section]
    \newtheorem{definition}{Definition}[section]
    
    \newtheorem{lemma}{Lemma}[section]
    \newtheorem{remark}{Remark}[section]

\begin{document}
\title{Ordering results between the largest claims arising from two general heterogeneous
portfolios}
\author{{\large { Sangita {\bf
                Das}\thanks {Email address: sangitadas118@gmail.com}~ and Suchandan {\bf Kayal}\thanks { Email address:
                kayals@nitrkl.ac.in
               }}} \\
    { \em \small {\it Department of Mathematics, National Institute of
            Technology Rourkela, Rourkela-769008, India.}}}
\date{}
\maketitle
\begin{center}
\noindent{\bf Abstract}
\end{center}
This work is entirely devoted to compare the largest claims from two
heterogeneous portfolios.  It is assumed that the
claim amounts in an insurance portfolio are
nonnegative absolutely continuous random
variables and belong to a general family of distributions. The largest claims have been compared based on various stochastic orderings. The established sufficient conditions are associated with the matrices and vectors of model parameters. Applications of the results are provided for the purpose of illustration.

\section{Introduction}

In survival analysis, models with nonmonotone
failure rate play a vital role to fit the real
life data sets. A large number of distributions
exists in statistical theory, which have monotone
failure rate. For example, the exponentiated
Weibull  and generalized gamma distributions have
monotone failure rate. In this communication, a
general family of distributions (exponentiated
location-scale) is taken.  It contains both
monotone and nonmonotone failure rate models.
Because of this, the general exponentiated
location-scale (ELS) model is important from both
practical and theoretical points of view.  It is
well-known that $X$ belongs to the ELS model if
$X\sim
F^{\alpha}(\frac{x-\lambda}{\theta}),~x>\lambda>0$
and $\alpha,\theta>0$. The functions $F(\cdot)$
and $f(\cdot)$ denote  the baseline cumulative
distribution and probability density functions of
$X$, respectively. Here, we consider $F(\cdot)$
to be the absolutely continuous distribution
function. The strictly positive real numbers
$\alpha,\lambda$ and $\theta$ are respectively
the shape, location and scale parameters. For
$\lambda=0$ and $\alpha=1$, the exponentiated
location-scale model reduces to the scale model.
Further, we respectively get the proportional
reversed hazard rate model and the location
model, when $\lambda=0$, $\theta=1$ and
$\alpha=1,$ $\theta=1$.

For $i=1,\cdots, n$, let $X_{i}$ be the claim
amount and $J_{i}$ be the Bernoulli random
variable. Further, $J_{i}=0$, if the $i$th policy
holder does not claim, and $J_{i}=1$, if the
$i$th policy holder makes random claim $X_{i}$.
We assume that the claim is taken place with
probability $p_{i}$, and is not taken place with
probability $1-p_{i}$. It is known that in an
insurance portfolio consisting of $n$ risks, the
$i$th individual risk is a product of $X_{i}$ and
$J_{i}$. Throughout this paper, we consider that
$\{X_{1},\cdots,X_{n}\}$ and
$\{Y_{1},\cdots,Y_{n}\}$ are two collections of
independent random claims of two portfolios with
$X_{i}\sim
F^{\alpha_{i}}(\frac{x-\lambda_{i}}{\theta_{i}})$
and $Y_{i}\sim
F^{\beta_{i}}(\frac{x-\mu_{i}}{\delta_{i}})$,
where $i=1,\cdots,n$. Further,  assume that
$\{J_{1},\cdots,J_{n}\}$ and
$\{J^{*}_{1},\cdots,J^{*}_{n}\}$ are another two
collections of independent Bernoulli random
variables, independent of  $X_{i}$'s and
$Y_{i}$'s, respectively, $i=1,\cdots,n$. Consider
two vectors $\bm{U}=(U_{1},\cdots,U_{n})$ and
$\bm{V}=(V_{1},\cdots,V_{n}),$ such that for
$i=1,\cdots,n,$  $U_{i}=J_{i}X_{i}$ and
$V_{i}=J_{i}^*Y_{i}$ with $E(J_{i})=p_{i}$ and
$E(J_{i}^*)=q_{i}$. Denote
$U_{n:n}=\max\{U_{1},\cdots,U_{n}\}$ and
$V_{n:n}=\max\{V_{1},\cdots,V_{n}\}$ for the
maximum claims arising from two insurance
portfolios of $n$ risks, where the $i$th
individual risks are $U_{i}$ and $V_{i}$,
respectively, $i=1,\cdots,n$.  Besides this,
$U_{n:n}$ has another interpretation in
reliability theory. It  represents the  lifetime
of  a parallel system for which the components
are equipped with starters. Here, the random
variables $X_{i}$'s can be treated as components'
lifetimes and $J_{i}$'s represent the status of
the corresponding starters.  Therefore, the
present study of stochastic comparison is very
important both from the mathematical research and
real life applications.

\cite{barmalzan2015convex} and
\cite{barmalzan2016likelihood} considered two
collections of independent  claims following
heterogeneous Weibull distributions.
\cite{barmalzan2015convex} addressed the
comparisons between the minimum claims
stochastically in the sense of  the convex
transform and right spread orders. They also
derived upper and lower bounds of the coefficient
of variations.  \cite{barmalzan2016likelihood}
discussed the sufficient conditions  under which
the likelihood ratio  and  dispersive orders hold
between the smallest claim amounts.
\cite{barmalzan2017ordering} took  scale model to
compare the extreme claims with respect to the
usual stochastic and the hazard rate orderings.
\cite{balakrishnan2018ordering} studied ordering
properties of the largest claim amounts from two
heterogeneous sets of portfolios. They proposed
sufficient conditions to show various stochastic
orderings between the largest claim amounts.
\cite{zhang2019ordering} established  conditions
to compare the extreme claims from two
collections of insurance portfolios. To the best
of our knowledge, stochastic comparisons of the
largest claim amounts when random claims have ELS
models have not been addressed in the literature
so far. However, some generalized models were
considered by \cite{das2019ordering},
\cite{das2019orderingmo} and \cite{dasordering}
to study ordering properties of extreme order
statistics in the context of reliability studies.
In this paper, we address this problem and derive
sufficient conditions for the stochastic
comparison of the largest claim amounts in the
sense of various stochastic orderings.

This article is organized as follows. In section
$2$, we provide some basic definitions and
results. Section $3$ is emphasized on some
ordering results based on the matrix chain
majorization order,  when heterogeneity presents
in two parameters. Section $4$ addresses
comparisons between the largest claims
with respect to the usual stochastic and reversed
hazard rate orders. Here, we consider that the
heterogeneity is presented in one parameter.
Section $5$ is devoted to illustrations of
the results.   Generalized linear failure rate
and Marshal-Olkin extended quasi
Lindley distributions are considered. Finally, we present some
concluding remarks in Section $6$.

Throughout, we assume that the random variables
are  nonnegative and absolutely continuous. The
integrations and differentiations are well
defined. Further, `increasing' and `decreasing'
terms are employed in non-strict sense. For any
function $h(.)$, $h'(x)=\frac{dh(x)}{dx}$. We
assume componentwise comparison when comparing
two vectors.

\section{Preliminaries\setcounter{equation}{0}}

This section is concerned with some basic
definitions and important lemmas, which are used
to prove the results in the subsequent sections.
Let $U$ and $V$ be two nonnegative absolutely
continuous random variables. Assume that
$f_{U}(.)$ and $f_{V}(.)$, $F_{U}(.)$ and
$F_{V}(.)$, $\bar{F}_{U}(.)$ and $\bar{F}_{V}(.)$
are the probability density functions, the
cumulative distribution functions and the
survival functions of $U$ and $V$, respectively.
The following definition is for some concepts of
stochastic orders. For comprehensive discussions
on the properties and applications of the
following stochastic orders, one may refer to
\cite{shaked2007stochastic}.
\begin{definition}\label{def2.1}
$U$ is smaller than $V$ in the
\begin{enumerate}
\item[(a)] usual stochastic order, abbreviated by $U\le_{st}V$, if $\bar{F}_{U}(x)
\le \bar{F}_{V}(x)$, for every $x\in \mathbb{R}$;
\item[(b)] reversed hazard rate order, abbreviated by $U\le_{rh}V$, if the ratio
$F_{V}(x)/F_{U}(x)$ is increasing with respect to $x$.
\end{enumerate}
\end{definition}

Next definition describes the concept of
majorization and related orders. Prior to this,
we consider vectors $\bm{x}=(x_{1},\cdots,x_{n})$
and $\bm{y}=(y_{1},\cdots,y_{n})$. Further,
 $x_{1:n}\le\cdots\le x_{n:n}$ and
$y_{1:n}\le\cdots\le y_{n:n}$ respectively denote
the components of $\bm{x}$ and $\bm{y}$.
\begin{definition}\label{def2.2}
$\bm{x}$ is known to be
\begin{enumerate}
\item[(a)] weakly supermajorized by $\bm{y}$, abbreviated by
$\bm{x}\preceq^{w}\bm{y}$,
if $\sum_{k=1}^{j}x_{k:n}\ge
\sum_{k=1}^{j}y_{k:n}$, for $j=1,\cdots,n$;

\item[(b)] weakly submajorized by $\bm{y}$, abbreviated by
$\bm{x}\preceq_{w}\bm{y}$,
if $\sum_{k=i}^{n}x_{k:n}\le
\sum_{k=i}^{n}y_{k:n}$, for $i=1,\cdots,n$;

\item[(c)] majorized by $\bm{y}$, abbreviated by
$\bm{x}\preceq^{m}\bm{y}$,
if $\sum_{k=1}^{n}x_{k}=\sum_{k=1}^{n}y_{k}$ and
$\sum_{k=1}^{j}x_{k:n}\ge \sum_{k=1}^{j}y_{k:n}$,
for $j=1,\cdots,n-1$;

\item[(d)] $p$-larger than $\bm{y}$, denoted by
            $\bm{x}\succeq^{p}\bm{y}$, if $\prod_{i=1}^{k}x_{i:n} \le
            \prod_{i=1}^{k}y_{i:n}$, for $k=1,\cdots,n$;

\item[(e)] reciprocally majorized by $\boldsymbol{y},$ denoted
        by $\boldsymbol{x}\preceq^{rm}\boldsymbol{y}$, if $\sum_{i=1}^{l}x_{i:n}^{-1}\le
        \sum_{i=1}^{l}y_{i:n}^{-1}$, for all $l=1,\ldots,n$.
\end{enumerate}
\end{definition}
It is easy to see that $\bm{x}\preceq^{m}\bm{y}$
implies both $\bm{x}\preceq^{w}\bm{y}$ and
$\bm{x}\preceq_{w}\bm{y}$. Further,
$\bm{x}\preceq^{w}\bm{y}\Rightarrow
\bm{x}\preceq^{p}\bm{y}$. We may refer to
\cite{Marshall2011} for brief and extensive
details on the majorization and their
applications.
\begin{definition}\label{def2.3}
Let
$\varphi:\mathbb{B}~(\subseteq{\mathbb{R}}^{n})\rightarrow
\mathbb{R}$ be a function. It is said to be Schur-convex
(Schur-concave) on $\mathbb{B}$ if
$\bm{x}\preceq^{m} \bm{y}\Rightarrow
\varphi(\bm{x})\leq (\geq) \varphi(\bm{y})$, for any
$\bm{x},\bm{y}\in \mathbb{B}.$
\end{definition}
The following notations are used throughout the
paper.  We denote $\bm{1}_{n}=(1,\cdots,1)$.
$$\mathcal{D}_{+}=\{(t_1,\cdots,t_n):t_{1}\geq\cdots\geq t_{n}>0\};~
\mathcal{E}_{+}=\{(t_1,\cdots,t_n):0<t_{1}\leq\cdots\leq t_{n}\}.$$

Now, we move our attention to  the notion of the
matrix majorization. We say that a square matrix
$\pi$ is said to be a permutation matrix, if each
row and column have a single entry, and except
that, all entries are zero. One can easily find
out that $n!$ number of such matrices arise after
interchanging rows (or columns) of the  $n\times
n$ order identity matrix $I_{n}$. Let $T_{w}$
denote a $T$-transform matrix, with the form
\begin{eqnarray}
T_{w}=wI_{n}+(1-w)\pi,~0\le w\le 1,
\end{eqnarray}
where $\pi$ is a permutation matrix obtained by
interchanging two rows or columns. Consider the
$T$-transform matrices $T_{w_{1}}$ and
$T_{w_{2}}$ such that
$T_{w_{1}}=w_{1}I_{n}+(1-w_{1})\pi_{1}$ and
$T_{w_{2}}=w_{2}I_{n}+(1-w_{2})\pi_{2},$ where
$\pi_{1}$ and $\pi_{2}$ are the permutation
matrices obtained by interchanging two rows or
columns and $~0\le w_{1},w_{2}\le 1$. We say that
$T_{w_{1}}$ and $T_{w_{2}}$ have the same
structure, if $\pi_{1}=\pi_{2}$, and have
different structures, if $\pi_{1}\neq\pi_{2}.$
The definition given below describes the concept
of multivariate majorization.

\begin{definition}\label{def2.4}
Let us take two matrices $C=[c_{ij}]$
and $D=[d_{ij}]$ of order $m\times n$, where $i=1,\cdots,m$ and
$j=1,\cdots,n$. Let $T_{w_{1}},\cdots,T_{w_{k}}$ be a finite set of $n\times n$ $T$-transform
matrices. Then,
$C$ is said to chain
majorize $D,$ abbreviated by $C>>D,$ if
$D=C T_{w_{1}}\cdots T_{w_{k}}.$
\end{definition}

Henceforth, $(\bm{r_{1}},\cdots,\bm{r_{m}};n)$
represents a  matrix of order $m\times n,$ where
each real-valued vector $\bm{r_{i}}$ containing
$n$ elements denotes the $i$th row, for
$i=1,\cdots,m.$ For $i, j=1,\cdots,n $ and
$x_{i}, y_{j} > 0 $, let us consider
 \begin{align*}\label{matrix2.1}
    M_{n} &= \left\{(\bm{x},\bm{y};n) :\
    (x_{i} - x_{j})(y_{i} - y_{j}) \geq 0\right\};\\
    Q_{n} &= \left\{(\bm{x},\bm{y};n) :\
    (x_{i} - x_{j})(y_{i} - y_{j}) \leq 0\right\}.
    \end{align*}

The next consecutive lemmas are
helpful to prove few of the proposed ordering
results. Interested readers are referred to the work by
\cite{balakrishnan2015stochastic} for detailed
idea on the proofs.
\begin{lemma}\label{lem2.2}
A differentiable function $\varpi :
\mathbb{R}^{+^{4}} \rightarrow \mathbb{R}^{+}$
satisfies
\begin{eqnarray}\label{eq2.2}
\varpi(C) \geq (\leq) \varpi(D) \text{ for all } C, D \in M_{2} \text{ or } Q_2~ \text{
and } C \gg D
\end{eqnarray}
if and only if for all $C\in M_{2}$ or $Q_2$ and $C$ satisfies
\begin{enumerate}
\item[(a)] $\varpi(C) = \varpi(C\pi)$, for all permutation matrices $\pi$ and
\item[(b)] $\sum_{i=1}^{2} \left(c_{ik} - c_{ij}\right)
\left( \varpi_{ik}(C) - \varpi_{ij}(C) \right)
\geq (\leq) 0$,  for all $j, k = 1, 2$,
\text{where} $\varpi_{ij}(C) = \frac{\partial
\varpi(C)}{\partial c_{ij}}$.
\end{enumerate}
\end{lemma}

\begin{lemma}\label{lem2.3}
Consider a differentiable function $\upsilon : \mathbb{R}^{+^{2}}
\rightarrow \mathbb{R}^{+}$
and a function $\zeta_{n} : \mathbb{R}^{+^{2n}}
\rightarrow \mathbb{R}^{+}$ such that
\begin{eqnarray}\label{eq2.3}
\zeta_{n}(C) = \prod_{k = 1}^{n} \upsilon \left(
c_{1k} , c_{2k}\right).
\end{eqnarray}
Let $\zeta_{2}$ satisfy
$(\ref{eq2.2})$. Then, for all $C\in M_{n} \text{ or } Q_n$ and $D= C T_{\omega}$, we have $\zeta_{n}(C)
\geq(\leq) \zeta_{n}(D)$.
\end{lemma}

Before proceeding to the next section, we
introduce two lemmas, which deal with the
analytical behavior of two mathematical
functions. The proofs are omitted since these are
simple. The first lemma is useful for the
derivation of the results of Theorems
\ref{th4.4}, \ref{th4.5}, \ref{th4.6}, and the
second one is used in the proof of Theorem
\ref{th3.1a}
\begin{lemma}\label{lem2.6}
    Consider a function $k_1:(0,\infty)\times (0,1)\times (0,1)\rightarrow (0,\infty)$ such that
    $k_1(\alpha,t,p)=\frac{(1-t^{\alpha})}{1-p(1-t^{\alpha})}$. Then, for all
    \begin{itemize}
        \item[(i)] $t\in(0,1)$, $k_1(\alpha,t,p)$ is increasing in
        to $\alpha$ and $p$;
        \item[(ii)] $\alpha\in(0,\infty)$ and $p\in(0,1)$, $k_1(\alpha,t,p)$ is decreasing
        in $t$.
    \end{itemize}
\end{lemma}

\begin{lemma}\label{lem2.7}
    Let $k_2:(0,\infty)\times (0,1)\rightarrow (0,\infty)$ be defined as
    $k_2(\alpha,p)=\frac{p t^{\alpha}\ln(t)}{1-p(1-t^{\alpha})}$. Then, for all
    \begin{itemize}
        \item[(i)] $p\in(0,1)$, $k_2(\alpha,p)$ is increasing in $\alpha;$
        \item[(ii)] $\alpha\in(0,\infty)$, $k_2(\alpha,p)$ is decreasing in $p$.
    \end{itemize}
\end{lemma}

\section{Matrix chain majorization \setcounter{equation}{0}}
In this section, we establish some ordering
results between the largest claims when a
matrix of parameters is related to another matrix of
parameters in some mathematical senses. To begin
with, let us write the respective cumulative distribution
functions of $U_{n:n}$ and $V_{n:n}$ as
\begin{eqnarray}\label{eq3.1}
F_{n:n}(t)=\prod_{i=1}^{n}\left[1-p_{i}\left[1-{F}^{\alpha_{i}}\left(\frac{t-\lambda_{i}}{\theta_{i}}\right)\right]\right],
~t>\max\{\lambda_1,\cdots,\lambda_n\}
\end{eqnarray}
and
\begin{eqnarray}\label{eq3.2}
G_{n:n}(t)=\prod_{i=1}^{n}\left[1-q_{i}\left[1-{F}^{\beta_{i}}\left(\frac{t-\mu_{i}}
{\delta_{i}}\right)\right]\right],
~t>\max\{\mu_1,\cdots,\mu_n\}.
\end{eqnarray}
First, we consider the following assumptions.
These will be called in the main results when
necessary.

\begin{itemize}
    \item[$(A1)$] Let $F(.)$ be the baseline distribution function. Again, let $\{X_{1},\cdots,X_{n}\}$ and $\{Y_{1},\cdots,Y_{n}\}$ be two sets of nonnegative
    independent random variables. For $i=1,\cdots,n,$ assume that $X_{i}\sim
    F^{\alpha_{i}}(\frac{x-\lambda_{i}}{\theta_{i}})$ and $Y_{i}\sim
    F^{\beta_{i}}(\frac{x-\mu_{i}}{\delta_{i}}).$
    \item[$(A2)$] Let $\{J_{1},\cdots,J_{n}\}$ and $\{J_{1}^*,\cdots,J_{n}^*\}$ be
two collections of independent Bernoulli random
variables,
    independent of $X_{i}$'s and $Y_{i}$'s, respectively. Further, $E(J_{i})=p_{i}$ and $E(J_{i}^*)=q_{i}$, for
    $i=1,\cdots,n$.
\end{itemize}

Let $r(.)$ be the hazard rate function of the baseline distribution $F(.),$ where $r(x)=f(x)/(1-F(x)).$ We provide the following conditions, which are also required for the smooth presentation of the results.
\begin{itemize}
    \item[$(C1)$]$r(x)$ is decreasing.
    \item[$(C2)$]$xr(x)$ is decreasing.
    \item[$(C3)$]$x^2r(x)$ is decreasing.
    \item[$(C4)$]$\frac{r'(x)}{r(x)}$ is increasing.
    \item[$(C5)$]$xr(x)$ is convex.
    \item[$(C6)$] $x^3r^{2}(x)$ is decreasing.
    \item[$(C7)$]$x^2[xr(x)]'$ is increasing.
    \item[$(C8)$]$r(x)$ is convex.
\end{itemize}
The function $\psi:(0,1)\rightarrow(0,\infty)$ is taken to be differentiable throughout the paper.
\begin{itemize}
    \item[$(C9)$]$\psi(w)$ is convex and increasing.
    \item[$(C10)$]$\psi(w)$ is convex and decreasing.

\end{itemize}
The following result establishes conditions, under which multivariate chain majorization between two matrices of parameters implies the usual stochastic order between the largest claim amounts. In the sequel, we take a common shape parameter vector for both sets. It is equal to a scalar $\alpha$, which lies in the interval $(0,1]$. Note that the following result contains two parts, of which the second part generalizes Theorem $1$ of \cite{barmalzan2017ordering}.

\begin{theorem}\label{th4.1}
For $n=2$, let $(A1), (A2)$ and $(C1)$ hold. Also, assume $\bm{\alpha}=\bm{\beta}=\alpha\bm{1}_2~(\alpha\le1) .$
\begin{itemize}
    \item[$(i)$] Suppose  the function  $\psi$  satisfies $(C9)$. If $\bm{\theta}=\bm{\delta}=\theta\bm{1}_2$ and $(\bm{\psi}(\bm{p}),\bm{\lambda};2)\in M_{2},$  then
  $(\bm{\psi}(\bm{p}),\bm{\lambda};2)\gg (\bm{\psi}(\bm{q}),\bm{\mu};2)
  \Rightarrow U_{2:2} \geq_{st} V_{2:2}$;
    \item[$(ii)$] Suppose the function  $\psi$ satisfies $(C10)$. If  $\bm{\lambda}=\bm{\mu}=\mu\bm{1}_2$ and $(\bm{\psi}(\bm{p}),1/\bm{\theta};2)\in M_{2},$
    then $(\bm{\psi}(\bm{p}),1/\bm{\theta};2)\gg (\bm{\psi}(\bm{q}),1/\bm{\delta};2)
  \Rightarrow U_{2:2} \geq_{st} V_{2:2}$.
\end{itemize}
\end{theorem}

\begin{proof} $(i)$ We have
\begin{eqnarray}\label{eq4.1}
F_{2:2}(t)=\prod_{i=1}^{2}\left[1-\psi^{-1}(w_{i})
\left[1-{F}^{\alpha}\left(\frac{t-\lambda_{i}}{\theta}\right)\right]\right],
\end{eqnarray}
where  $\psi(p_{i})=w_{i,}$ for $i=1,2$. Note that
$F_{2:2}(t)$ is permutation invariant
in $(w_{i},\lambda_{i})$. Hence, the first condition
of Lemma \ref{lem2.2} is fulfilled. Further, the partial derivatives of \eqref{eq4.1}
with respect to $w_{i}$ and $\lambda_{i}$ are respectively obtained as

\begin{eqnarray}\label{eq4.2}
\frac{\partial F_{2:2}(t)}{\partial
w_{i}}=-\frac{\partial
\psi^{-1}(w_{i})}{\partial w_{i}}\frac{[1-{F^{\alpha}\left(\frac{t-\lambda_{i}}{\theta}\right)}]}{\left[1-\psi^{-1}(w_i)\left[1-{F^{\alpha}\left(\frac{t-\lambda_{i}}{\theta}\right)}\right]\right]}
F_{2:2}(t),\end{eqnarray}
and
 \begin{eqnarray}\label{eq4.2.}
\frac{\partial F_{2:2}(t)}{\partial
\lambda_{i}}=-\frac{\alpha {F^{\alpha-1}\left(\frac{t-\lambda_{i}}{\theta}\right)}f\left(\frac{t-\lambda_{i}}{\theta}\right)}{\theta\left[1-{F^{\alpha}\left(\frac{t-\lambda_{i}}{\theta}\right)}\right]}\frac{\psi^{-1}(w_i)[1-{F^{\alpha}\left(\frac{t-\lambda_{i}}{\theta}\right)}]}{1-\psi^{-1}(w_i)\left[1-{F^{\alpha}\left(\frac{t-\lambda_{i}}{\theta}\right)}\right]}F_{2:2}(t).
\end{eqnarray}
We define
\begin{eqnarray}\label{eq4.3}
\phi_1(\bm{w},\bm{\lambda})=(w_i-w_j)\left[\frac{\partial
F_{2:2}(t)}{\partial
w_{i}}-\frac{\partial
F_{2:2}(t)}{\partial
w_{j}}\right]+(\lambda_{i}-\lambda_{j})\left[\frac{\partial
F_{2:2}(t)}{\partial
\lambda_{i}}-\frac{\partial
F_{2:2}(t)}{\partial
\lambda_{j}}\right],
\end{eqnarray}
where the partial derivatives are given by \eqref{eq4.2} and \eqref{eq4.2.}. Together with Lemma $3$ of \cite{balakrishnan2015stochastic} and the assumptions made, we can show that $\phi_1(\bm{w},\bm{\lambda})$ is nonpositive. Thus, clearly, the second argument of Lemma \ref{lem2.2} is verified, and the proof is completed. By adopting the arguments of the proof of the first part, the second part follows easily.
\end{proof}

\begin{remark}
On using Theorem \ref{th4.1}(ii), one can easily find out a lower bound for the reliability function of the largest claims having heterogeneous portfolios of risks in the form of  reliability of the largest claims having homogeneous portfolio of risks.
Consider a $T$- transform matrix $T_{0.5}$ of order  $2\times2,$  where the first and second rows are same and equal to $(1/2,1/2)$. Let $({\psi}(p_1),{\psi}(p_2))=(e^{-p_1},e^{-p_2})$ and $1/\boldsymbol{\theta}=(1/\theta_1,1/\theta_2)$. Further, assume $({\psi}(q_1),{\psi}(q_2))=\left((e^{-p_1}+e^{-p_2})/2, (e^{-p_1}+e^{-p_2})/2\right)$  and $(1/\delta_1,1/\delta_2)=\left((\theta_1+\theta_2)/(2\theta_1\theta_2),(\theta_1+\theta_2)/(2\theta_1\theta_2)\right)$. It is easy to check that $$\left( \begin{smallmatrix} \psi(q_1)&\psi(q_2)\\ 1/\delta_1 & 1/\delta_2\\
\end{smallmatrix} \right)=\left( \begin{smallmatrix} \psi(p_1)&\psi(p_2)\\ 1/\theta_1 &1/ \theta_2\\
\end{smallmatrix} \right)T_{0.5}.$$ As a result, $\left( \begin{smallmatrix} \psi(p_1)&\psi(p_2)\\  1/\theta_1 &1/ \theta_2\\
\end{smallmatrix} \right)\gg \left( \begin{smallmatrix} \psi(q_1)&\psi(q_2)\\ 1/\delta_1 & 1/\delta_2\\
\end{smallmatrix} \right).$ Therefore, by Theorem \ref{th4.1}$(ii)$, we can propose a lower bound of the reliability function of $U_{2:2}$ as
 $$ \bar{F}_{2:2}(t)\geq 1-\left\{1+\ln\left(\frac{e^{-p_{1}}+e^{-p_{2}}}{2}\right)\left[1-{F}^{\alpha}
 \left(\frac{(t-\lambda)(\theta_1+\theta_2)}{2\theta_1\theta_2}\right)\right]\right\}^2. $$
\end{remark}

Now, it might be of interest to investigate whether the decreasing and convexity property of the function $\psi$ is a must or not in Theorem \ref{th4.1}$(ii)$. The following numerical counterexample states that this assumption is required to get the usual stochastic order in Theorem \ref{th4.1}$(ii)$.

\begin{counterexample}\label{ce3.1}
    Let $F(t)=1-(1+t^5)^{-4},~t> 0$ and $\psi(p)=1-p^3.$ It is easy to check that $\psi(p)$ is decreasing and concave. So, the condition in $(C10)$ is relaxed. Let us take $(\alpha_1,\alpha_2)=(\beta_1,\beta_2)=(0.01,0.01),$ $(1/\theta_1,1/\theta_2)=(0.7,0.6),$ $(1/\delta_1,1/\delta_2)=(0.66,0.64)$, $(\lambda_1,\lambda_2)=(\mu_1,\mu_2)=(0.9,0.9)$,  $(p_1,p_2)=((0.2)^{1/3},(0.5)^{1/3})$ and $(q_1,q_2)=((0.32)^{1/3},(0.38)^{1/3})$. Consider a $T$-transform matrix $T_{0.6}=\left( \begin{smallmatrix} 0.6 & 0.4\\ 0.4 &0.6\\
    \end{smallmatrix} \right)$. It can be shown that
    $$\left( \begin{smallmatrix} \psi(q_1)&\psi(q_2)\\ 1/\delta_1 & 1/\delta_2\\
    \end{smallmatrix} \right)=\left( \begin{smallmatrix} \psi(p_1)&\psi(p_2)\\ 1/\theta_1 &1/ \theta_2\\
    \end{smallmatrix} \right)T_{0.6},$$
    which produces  $\left( \begin{smallmatrix} \psi(p_1)&\psi(p_2)\\  1/\theta_1 &1/ \theta_2\\
    \end{smallmatrix} \right)\gg \left( \begin{smallmatrix} \psi(q_1)&\psi(q_2)\\ 1/\delta_1 & 1/\delta_2\\
    \end{smallmatrix} \right)$. Further, $\left( \begin{smallmatrix} \psi(p_1)&\psi(p_2)\\  1/\theta_1 &1/ \theta_2\\
    \end{smallmatrix} \right)\in M_{2}$.
Based on this present setup, we have
$F_{2:2}(1.5)-G_{2:2}(1.5)=1.5928\times
e^{-5}~(>0)$ and  $
F_{2:2}(1.6)-G_{2:2}(1.6)=-1.9324\times
e^{-4}~(<0)$. For graphical view, please see
Figure $1(a)$.
    It establishes that Theorem \ref{th4.1}$(ii)$ does not hold, if $(C10)$ is taken out.
\end{counterexample}
\begin{figure}[h]
    \begin{center}
    \subfigure[]{\label{c1}\includegraphics[height=2.41in]{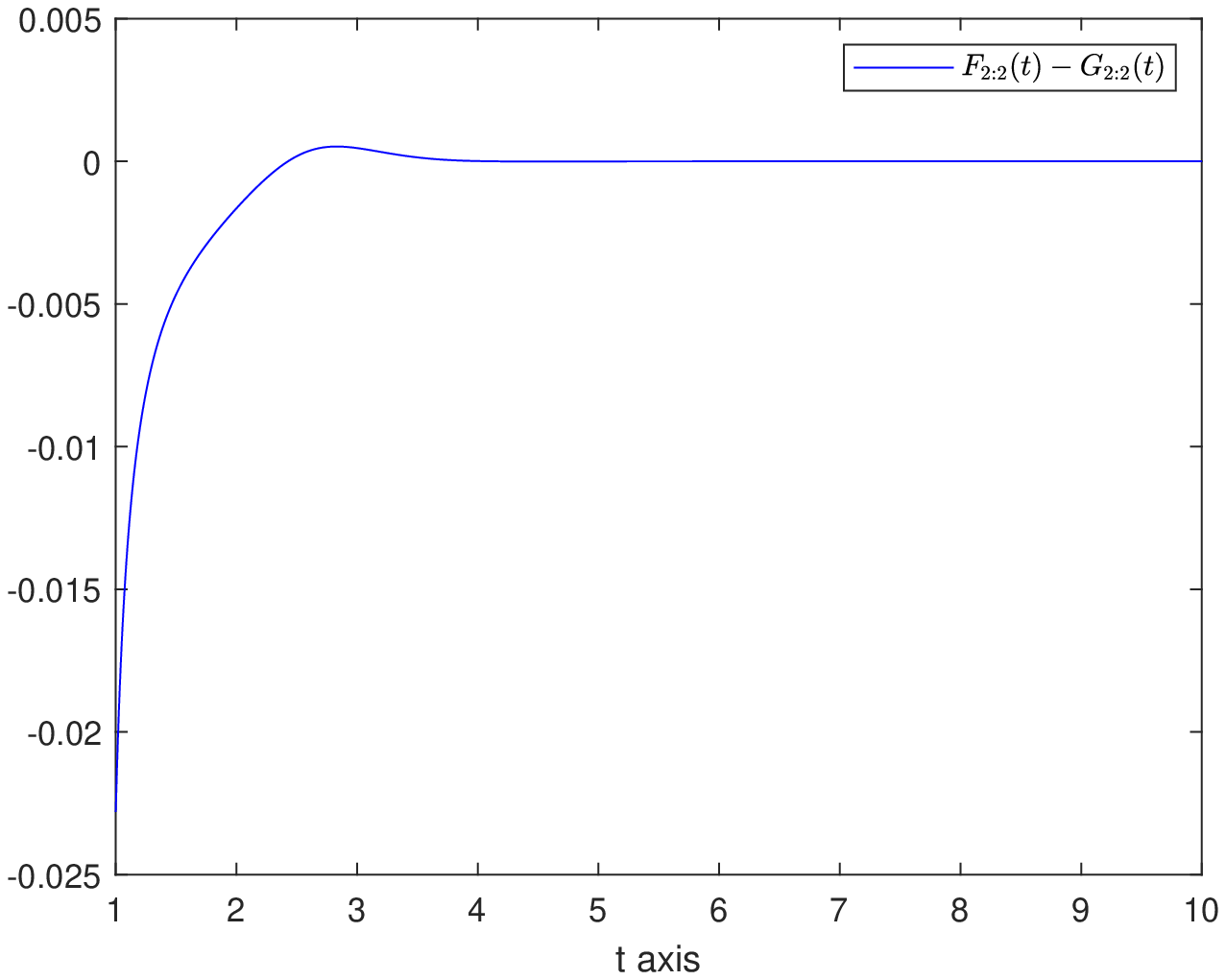}}
    \subfigure[]{\label{c2}\includegraphics[height=2.41in]{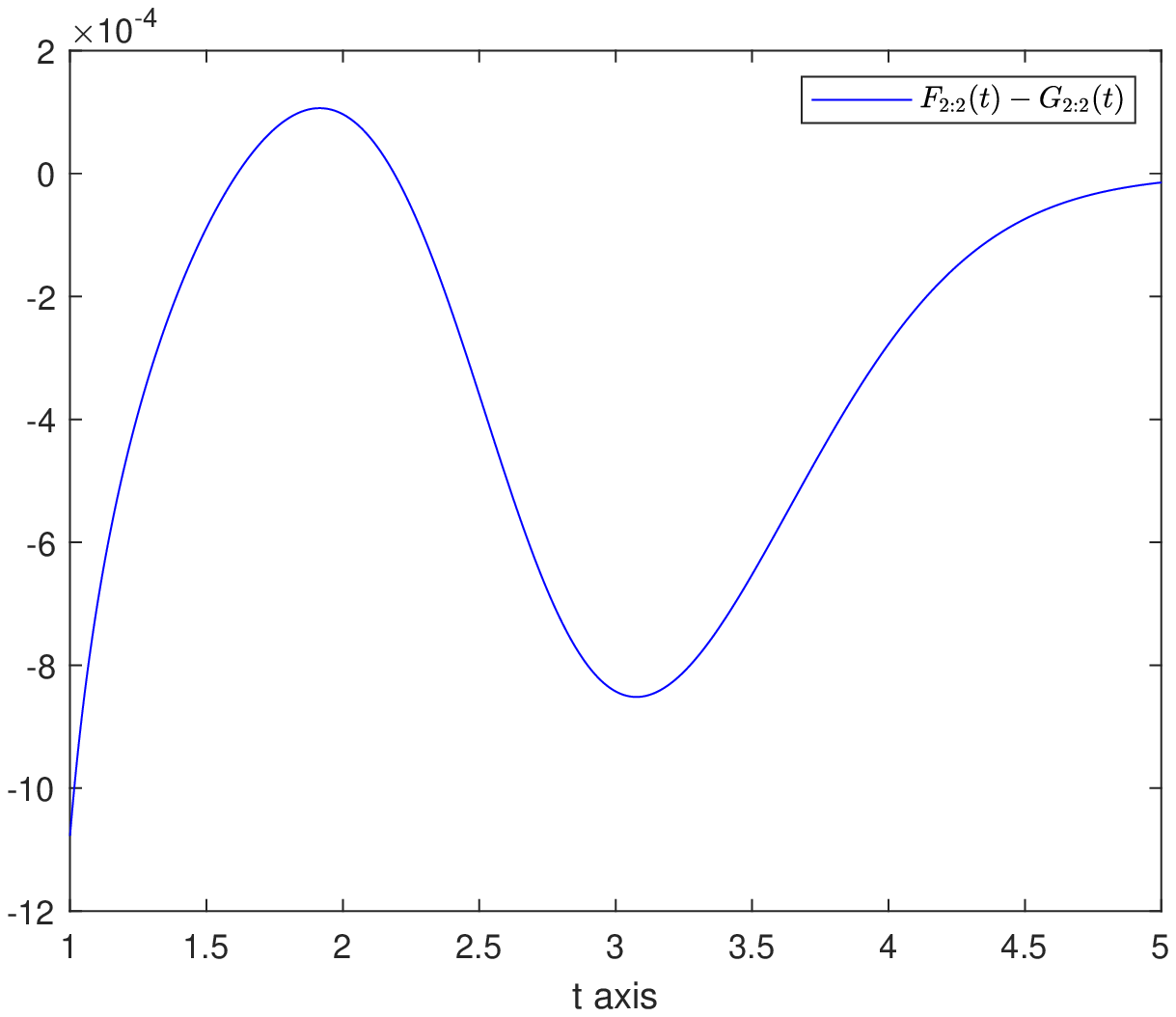}}
    \caption{(a) Graph of $F_{2:2}(t)-G_{2:2}(t)$ for  Counterexample \ref{ce3.1}. (b)  Graph of $F_{2:2}(t)-G_{2:2}(t)$ for  Counterexample \ref{ce3.2}. }
    \end{center}
    \end{figure}

We next consider a counterexample to show that if the condition ``$\left( \begin{smallmatrix} \psi(p_1)&\psi(p_2)\\  1/\theta_1 &1/ \theta_2\\
\end{smallmatrix} \right)\notin M_{2}$" is removed, then Theorem \ref{th4.1}$(ii)$ may not hold.

\begin{counterexample}\label{ce3.2}
Consider the same baseline distribution function
as in Counterexample \ref{ce3.1}. Suppose
$\psi(p)=-\ln p.$ Here, one can check that $(C1)$
and $(C10)$ are satisfied. Set
$(\alpha_1,\alpha_2)=(\beta_1,\beta_2)=(0.01,0.01),$
$(1/\theta_1,1/\theta_2)=(0.5,0.3),~(1/\delta_1,1/\delta_2)=(0.32,0.48)$,
$(\lambda_1,\lambda_2)=(\mu_1,\mu_2)=(0.9,0.9)$,
$(\psi(p_1),\psi(p_2))=(0.23,0.69)$ and
$(\psi(q_1),\psi(q_2))=(0.644,0.276)$.
    For a $T$-transform matrix $T_{0.1}=\left( \begin{smallmatrix} 0.1 & 0.9\\ 0.9 &0.1\\
    \end{smallmatrix} \right)$, we have
    $$\left( \begin{smallmatrix} \psi(q_1)&\psi(q_2)\\ 1/\delta_1 & 1/\delta_2\\
    \end{smallmatrix} \right)=\left( \begin{smallmatrix} \psi(p_1)&\psi(p_2)\\ 1/\theta_1 &1/ \theta_2\\
    \end{smallmatrix} \right)T_{0.1}.$$
    Thus, from Definition \ref{def2.4}, $\left( \begin{smallmatrix} \psi(p_1)&\psi(p_2)\\  1/\theta_1 &1/ \theta_2\\
    \end{smallmatrix} \right)\gg \left( \begin{smallmatrix} \psi(q_1)&\psi(q_2)\\ 1/\delta_1 & 1/\delta_2\\
    \end{smallmatrix} \right)$. Again, $\left( \begin{smallmatrix} \psi(p_1)&\psi(p_2)\\  1/\theta_1 &1/ \theta_2\\
    \end{smallmatrix} \right)$ does not belong to $ M_{2}$.
Now, considering the assumed numerical values, we
get $ F_{2:2}(1.6)-G_{2:2}(1.6)=-9.2138\times
e^{-6}~(<0)$ and  $
F_{2:2}(1.7)-G_{2:2}(1.7)=5.0560\times
e^{-5}~(>0)$, which negates the usual stochastic
order, stated in Theorem \ref{th4.1}$(ii)$. See
Figure $1(b)$ for the graph of
$F_{2:2}(t)-G_{2:2}(t).$
    \end{counterexample}

The following three consecutive results can be
thought of a generalization of the above result
to arbitrary $n\ge 3$. The cases in which the
chain majorization order holds between two
matrices   $(\bm{\psi}(\bm{q}),\bm{\mu};n)$ and
$(\bm{\psi}(\bm{p}),\bm{\lambda};n)$, or
$(\bm{\psi}(\bm{q}),1/\bm{\delta};n)$ and
$(\bm{\psi}(\bm{p}),1/\bm{\theta};n)$ are
considered. First, we use a $T$-transform matrix
to obtain the matrices
$(\bm{\psi}(\bm{q}),\bm{\mu};n)$ and
$(\bm{\psi}(\bm{q}),1/\bm{\delta};n)$ from
$(\bm{\psi}(\bm{p}),\bm{\lambda};n)$ and
$(\bm{\psi}(\bm{p}),1/\bm{\theta};n)$,
respectively, so that
$(\bm{\psi}(\bm{p}),\bm{\lambda};n)\gg
(\bm{\psi}(\bm{q}),\bm{\mu};n)$ and
$(\bm{\psi}(\bm{p}),1/\bm{\theta};n)\gg
(\bm{\psi}(\bm{q}),1/\bm{\delta};n)$ hold.

\begin{theorem}\label{th4.2}
     Let us assume that $(A1), (A2)$ and $(C1)$ hold. Further, $\bm{\alpha}=\bm{\beta}=\alpha \bm{1}_n~(\alpha\le 1).$ Suppose $T_{w}$ is a $T$-transform matrix.

\begin{itemize}
    \item[$(i)$] Let $\psi$ satisfy $(C9)$.  If $\bm{\theta}=\bm{\delta}=
    \theta\bm{1}_n$ and $(\bm{\psi}(\bm{p}),\bm{\lambda};n)\in M_{n}$, then
    $(\bm{\psi}(\bm{q}),\bm{\mu};n)=(\bm{\psi}(\bm{p}),\bm{\lambda};n)T_{w}
  \Rightarrow U_{n:n} \geq_{st} V_{n:n}$;

  \item[$(ii)$] Let $\psi$ satisfy $(C10)$. If  $\bm{\lambda}=\bm{\mu}=\mu\bm{1}_n$ and $(\bm{\psi}(\bm{p}),1/\bm{\theta};n)\in M_{n}$, then
    $(\bm{\psi}(\bm{q}),1/\bm{\delta};n)=(\bm{\psi}(\bm{p}),1/\bm{\theta};n)T_{w}
  \Rightarrow U_{n:n} \geq_{st} V_{n:n}$.
\end{itemize}
\end{theorem}

\begin{proof}
Here, we present the proof of
the first part. The second part can be proved with the similar arguments. Denote
$\zeta_{n}(\bm{w},\bm{\lambda})=\prod_{i=1}^{n}\left[1-\psi^{-1}(w_{i})\left[1-{F}^{\alpha}\left(\frac{t-\lambda_i}
{\theta}\right)\right]\right]=\prod_{i=1}^{n}\zeta_{i}({w_i},{\lambda_i})$, where
$\zeta_{i}({w_i},{\lambda_i})=1-\psi^{-1}(w_{i})\left[1-{F}^{\alpha}\left(\frac{t-\lambda_i}{\theta}\right)\right]$,
for $i=1,\cdots,n$. By Theorem
\ref{th4.1}$(i)$, one can easily check that
$\zeta_{2}$ satisfies (\ref{eq2.2}). Rest of the proof is completed from
Lemma
\ref{lem2.3}.
\end{proof}

Suppose $T_{w_{1}},\cdots,T_{w_{k}}$  are $k$
(finite) number of $T$-transform matrices. Also,
assume that all these $T$-transform matrices have
the same structure. Since, the product of a
finite number of $T$-transform matrices having
the same structure is again a $T$-transform
matrix. Hence, the above result can be extended.
This is addressed in the following corollary,
which is a direct consequence of Theorem
\ref{th4.2}. Here,
$(\bm{\psi}(\bm{p}),\bm{\lambda};n)\gg
(\bm{\psi}(\bm{q}),\bm{\mu};n)$ and
$(\bm{\psi}(\bm{p}),1/\bm{\theta};n)\gg
(\bm{\psi}(\bm{q}),1/\bm{\delta};n)$ hold, since
$(\bm{\psi}(\bm{q}),\bm{\mu};n)$ and
$(\bm{\psi}(\bm{q}),1/\bm{\delta};n)$ are
respectively obtained from
$(\bm{\psi}(\bm{p}),\bm{\lambda};n)$ and
$(\bm{\psi}(\bm{p}),1/\bm{\theta};n)$ using a
finite number of $T$-transform matrices. Denote
$T^{*}_{w}=T_{w_{1}}\cdots T_{w_{k}}.$

\begin{corollary}\label{th4.2a}
     Suppose  $(A1), (A2)$ and $(C1)$ hold. Also, let
     $\bm{\alpha}=\bm{\beta}=\alpha \bm{1}_n~(\alpha\le 1).$
        \begin{itemize}
        \item[$(i)$] Let    $\psi$ satisfy $(C9)$.  If $\bm{\theta}=\bm{\delta}=\theta\bm{1}_n$ and  $(\bm{\psi}(\bm{p}),\bm{\lambda};n)\in M_{n}$, then
        $(\bm{\psi}(\bm{q}),\bm{\mu};n)=(\bm{\psi}(\bm{p}),\bm{\lambda};n)T^{*}_{w}
        \Rightarrow U_{n:n} \geq_{st} V_{n:n}$;

        \item[$(ii)$] Let   $\psi$ satisfy $(C10)$.
        If  $\bm{\lambda}=\bm{\mu}=\mu\bm{1}_n$ and $(\bm{\psi}(\bm{p}),1/\bm{\theta};n)\in M_{n}$,
        then
        $(\bm{\psi}(\bm{q}),1/\bm{\delta};n)=(\bm{\psi}(\bm{p}),1/\bm{\theta};n)T^{*}_{w}
        \Rightarrow U_{n:n} \geq_{st} V_{n:n}$.
            \end{itemize}
\end{corollary}

Note that a finite product of $T$-transform
matrices may not produce a $T$-transform matrix,
when they have different structures. Thus, it is
natural to face the question whether Corollary
\ref{th4.2a} still holds for differently
structured $T$-transform matrices. The next
result shows that it is possible under some new
assumptions.

\begin{theorem}\label{th4.3}

Let $(A1), (A2)$ and $(C1)$ hold. Assume
$\bm{\alpha}=\bm{\beta}=\alpha
\bm{1}_n~(\alpha\le 1).$ Further, let
$T_{w_{1}},\cdots,T_{w_{j}}$ be the $T$-transform
matrices with different structures for
$j=1,\cdots,k-1,$ ($k \geq 2$).
\begin{itemize}
    \item[$(i)$] Suppose $\psi$ satisfies $(C9)$. If $\bm{\theta}=\bm{\delta}=\theta\bm{1}_n$ and
$(\bm{\psi}(\bm{p}),\bm{\lambda};n)\in M_{n}$,
then
$(\bm{\psi}(\bm{q}),\bm{\mu};n)=(\bm{\psi}(\bm{p}),\bm{\lambda};n)T_{w_{1}}\cdots
T_{w_{k}}\Rightarrow U_{n:n}\geq_{st} V_{n:n}$;

    \item[$(ii)$] Suppose   $\psi$ satisfies $(C10)$. If
$\bm{\lambda}=\bm{\mu}=\mu\bm{1}_n$ and
$(\bm{\psi}(\bm{p}),1/\bm{\theta};n)\in M_{n},$
then
$(\bm{\psi}(\bm{q}),1/\bm{\delta};n)=(\bm{\psi}(\bm{p}),1/\bm{\theta};n)T_{w_{1}}\cdots
T_{w_{k}}$
    $\Rightarrow U_{n:n}\geq_{st} V_{n:n}$.

\end{itemize}
\end{theorem}
\begin{proof}
We provide the proof of the first part. The
second part follows similarly. Let us consider
$(\bm{\psi}(\bm{p}^{(j)}),\bm{\lambda}^{(j)};n)=(\bm{\psi(p)},\bm{\lambda};n)T_{w_{1}}\cdots
$ $T_{w_{j}}$, where $j=1,\cdots,k-1.$ Further,
let $W_1^{(j)},\cdots,W_n^{(j)}$ be $n$
independent random variables having the
distribution function of  $W_i^{(j)}$ as $$
F^{(j)}_{W_{i}}(t)=1-\psi^{-1}(w^{(j)}_i)\left[1-F^{\alpha}\left(\frac{t-\lambda^{(j)}_{i}}{\theta}\right)
\right],$$ for $i=1,\cdots,n$ and
$j=1,\cdots,k-1.$ Using the assumptions made, we
have
$(\bm{\psi}(\bm{p}^{(j)}),\bm{\lambda}^{(j)};n)\in
M_{n},$ for $j=1,\cdots,k-1.$ Now,
\begin{align*}
\left(\bm{\psi(q)},\bm{\mu};n\right)&=\left(\bm{\psi}(\bm{p}),\bm{\lambda};n\right)T_{w_{1}}\cdots T_{w_{k}}\\
&=\left[\left(\bm{\psi}(\bm{p}),\bm{\lambda};n\right)T_{w_{1}}\cdots T_{w_{k-1}}\right]T_{w_{k}}\\
&=\left(\bm{\psi}(\bm{p}^{(k-1)}),\bm{\lambda}^{(k-1)};n\right)T_{w_{k}},
    \end{align*}
which implies $W^{(k-1)}_{n:n}\geq_{st} V_{n:n}$.
Similarly,
$(\bm{\psi}(\bm{p}^{(k-1)}),\bm{\lambda}^{(k-1)};n)=(\bm{\psi}(\bm{p}^{(k-2)})
,\bm{\lambda}^{(k-2)};n)$ $ T_{w_{k-1}}$ implies
$W^{(k-2)}_{n:n}\geq_{st} W^{(k-1)}_{n:n}$.
Applying similar arguments, we obtain
$$ U_{n:n}\geq_{st}W^{(1)}_{n:n}\geq_{st}\cdots \geq_{st} W^{(k-2)}_{n:n}\geq_{st}W^{(k-1)}_{n:n}\geq_{st} V_{n:n}.$$ This completes the proof of the first part.
 \end{proof}

Till now, we have derived conditions under which
the usual stochastic order holds between the
largest claim amounts arising from two
heterogeneous portfolios of risks. Now, one may
wonder whether the above results can be upgraded
to other stochastic orders. Below, we answer this
question affirmatively. In this part, it is shown
that under some new conditions, the usual
stochastic order can be extended to the reversed
hazard rate order. The reversed hazard rate of
$U_{n:n}$ is given by
\begin{eqnarray}\label{eq3.7}
\tilde{r}_{n:n}(t)=\sum_{i=1}^{n}\frac{\alpha_{i}}{\theta_i}r
\left(\frac{t-\lambda_i}{\theta_i}\right)\left[
\frac{p_iF^{(\alpha_{i}-1)}\left(\frac{t-\lambda_i}
    {\theta_i}\right)\left[1-F\left(\frac{t-\lambda_i}{\theta_i}\right)\right]}
{1-p_i\left[1-F^{\alpha_{i}}\left(\frac{t-\lambda_i}
    {\theta_i}\right)\right]}\right].
\end{eqnarray}
 The reversed hazard rate of  $V_{n:n}$,
denoted by $\tilde{s}_{n:n}(.)$ can be obtained
by substituting
$q_{i},\delta_{i},\mu_{i},\beta_{i}$ in the place
of $p_{i},\theta_{i}, \lambda_{i},\alpha_{i}$,
respectively in (\ref{eq3.7}). First, we consider
two heterogeneous insurance portfolios each
consisting of two individual risks.
\begin{theorem}\label{th4.4}
    For $n=2$, let $(A1)$ and $(A2)$ hold. Further, take $\bm{\alpha}=\bm{\beta}=\bm{1}_n.$
\begin{itemize}
    \item[$(i)$] If $\bm{\theta}=\bm{\delta}~(=\theta\bm{1}_2)$, $\bm{p}=\bm{q}~(=p\bm{1}_2)$ and $(\bm{\lambda},1/\bm{\theta};2)\in Q_{2}$, then we have $(\bm{\lambda},1/\bm{\theta};2)\gg (\bm{\mu},1/\bm{\delta};2)$
    $\Rightarrow U_{2:2} \geq_{rh} V_{2:2}$, provided $(C2)$, $C(3)$, $(C4)$ and $(C5)$ are satisfied;
 \item[$(ii)$] Let $\psi$ satisfy $(C9)$. If $\bm{\theta}=\bm{\delta}~(=\theta\bm{1}_2)$ and $(\bm{\lambda},\bm{\psi}(\bm{p});2)\in M_2$, then
$(\bm{\lambda},\bm{\psi}(\bm{p});2)\gg
(\bm{\mu},\bm{\psi}(\bm{q});2)\Rightarrow U_{2:2}
\geq_{rh} V_{2:2}$, provided $(C1)$ and $(C8)$
hold;
\item[$(iii)$] Let $\psi$ satisfy $(C9)$. If $\bm{\lambda}=\bm{\mu}~(=\lambda\bm{1}_2)$ and $(1/\bm{\theta},\bm{\psi}(\bm{p});2)\in Q_{2}$, then
$(1/\bm{\theta},\bm{\psi}(\bm{p});2)\gg (1/\bm{\delta},\bm{\psi}(\bm{q});2)\Rightarrow U_{2:2} \geq_{rh} V_{2:2}$, provided $(C2)$ and $(C5)$are satisfied.
\end{itemize}
\end{theorem}
\begin{proof} $(i)$ Under the set up, Equation \eqref{eq3.7} can be expressed as
    \begin{eqnarray}\label{eq4.5}
        \tilde{r}_{2:2}(t)=\sum_{i=1}^{2}{m_{i}}r
        \left((t-\lambda_i){m_i}\right)\left[
        \frac{p\left[1-F\left({(t-\lambda_i)}{m_i}\right)\right]}
        {1-p\left[1-F\left({(t-\lambda_i)}
            {m_i}\right)\right]}\right],
    \end{eqnarray} where $m_i=1/ \theta_i$, for $i=1,\cdots,n.$
    On differentiating (\ref{eq4.5}) with respect to
    $m_{i}$ and $\lambda_i$ partially, we respectively get
    \begin{eqnarray}\label{eq4.6}
        \frac{\partial \tilde{r}_{2:2}(t)}{\partial
            m_{i}}=&\frac{\partial}{\partial x}\left[{xr(x)}\right]_{x=\left((t-\lambda_i){m_i}\right)}\left[
        \frac{p\left[1-F\left((t-\lambda_i){m_i}\right)\right]}
        {1-p\left[1-F\left((t-\lambda_i){m_i}\right)\right]}\right]\\
        &-\left[xr^2(x)\right]_{x=\left((t-\lambda_i){m_i}\right)}\left[
        \frac{p\left[1-F\left((t-\lambda_i){m_i}\right)\right]}
        {\left[1-p\left[1-F\left((t-\lambda_i){m_i}\right)\right]\right]^2}\right]\nonumber
    \end{eqnarray}
and
\begin{eqnarray}\label{eq4.7}
    \frac{\partial \tilde{r}_{2:2}(t)}{\partial
         \lambda_{i}}=&-\frac{\left[x^2r(x)\right]_{x=\left((t-\lambda_i){m_i}\right)}}{(t-\lambda_i)^2}\left[\frac{r'(x)}{r(x)}\right]_{x=\left((t-\lambda_i){m_i}\right)}\left[
    \frac{p\left[1-F\left((t-\lambda_i){m_i}\right)\right]}
    {1-p\left[1-F\left((t-\lambda_i){m_i}\right)\right]}\right]\\
    &+\frac{1}{(t-\lambda_i)^2}\left[xr(x)\right]^2_{x=\left((t-\lambda_i){m_i}\right)}\left[
    \frac{p\left[1-F\left((t-\lambda_i){m_i}\right)\right]}
    {\left[1-p\left[1-F\left((t-\lambda_i){m_i}\right)\right]\right]^2}\right]\nonumber.
\end{eqnarray}
Now, consider the following
\begin{eqnarray}\label{eq4.4}
    \phi_2(\bm{m},\bm{\lambda})=(m_i-m_j)\left[\frac{\partial
        \tilde{r}_{2:2}(t)}{\partial
        m_{i}}-\frac{\partial
        \tilde{r}_{2:2}(t)}{\partial
        m_{j}}\right]+(\lambda_{i}-\lambda_{j})\left[\frac{\partial
        \tilde{r}_{2:2}(t)}{\partial
        \lambda_{i}}-\frac{\partial
        \tilde{r}_{2:2}(t)}{\partial
        \lambda_{j}}\right].
\end{eqnarray}
Under the assumptions and Lemma \ref{lem2.6}, it
can be shown that the right hand side of
\eqref{eq4.4} is greater than or equals to zero.
Thus, the result follows from Lemma \ref{lem2.2}.
Using similar approach, other two parts can be
proved.
 \end{proof}

The following counterexample shows that the
condition in $(C9)$ is necessary to obtain the
reversed hazard rate order in Theorem
\ref{th4.4}$(ii)$.
\begin{counterexample}\label{ce3.3}
Consider the baseline distribution function as $F(t)=1-(1+5t)^{-1/5},~t$
$> 0$, for which the hazard rate function $r(x)$ is decreasing and convex.
Let $\psi(p)=1-p^2.$ Clearly, $\psi(p)$ does not satisfy $(C9)$. Set $(\alpha_1,\alpha_2)=(\beta_1,\beta_2)=(1,1),$ $(\theta_1,\theta_2)=(\delta_1,\delta_2)=(0.5,0.5)$, $(\lambda_1,\lambda_2)=(0.9,0.6),~(\mu_1,\mu_2)=(0.81,0.69)$,  $(p_1,p_2)=(\sqrt{0.2},\sqrt{0.3})$ and $(q_1,q_2)=(\sqrt{0.23},\sqrt{0.27})$. Consider a $T$-transform matrix $T_{0.7}=\left( \begin{smallmatrix} 0.7 & 0.3\\ 0.3 &0.7\\
    \end{smallmatrix} \right)$. It can be shown that
    $$\left( \begin{smallmatrix}  \mu_1 & \mu_2\\\psi(q_1)&\psi(q_2)\\
    \end{smallmatrix} \right)=\left( \begin{smallmatrix}  \lambda_1 & \lambda_2\\\psi(p_1)&\psi(p_2)\\
    \end{smallmatrix} \right)T_{0.7},$$
    which yields $\left( \begin{smallmatrix} \lambda_1 & \lambda_2\\\psi(p_1)&\psi(p_2)\\
    \end{smallmatrix} \right)\gg \left( \begin{smallmatrix} \mu_1 & \mu_2\\\psi(q_1)&\psi(q_2)\\
    \end{smallmatrix} \right)$. Further, $\left( \begin{smallmatrix} \lambda_1 & \lambda_2\\\psi(p_1)&\psi(p_2)\\
    \end{smallmatrix} \right)\in M_{2}$ . Denote $\eta(t)=\tilde{r}_{2:2}(t)-\tilde{s}_{2:2}(t)$.
    Then, $\eta(1.7)=1.3522\times e^{-4}~(>0)$ and
    $\eta(1.8)=-2.0998\times e^{-4}~(<0)$, which shows that $\eta(t)$ changes
    sign, when $t$ travels from $0$ to $\infty$. Graph is presented in
    Figure $2(a)$ for clear view. Thus, $U_{2:2}\ngeq_{rh} V_{2:2}.$
     \begin{figure}[h]
        \begin{center}

            \subfigure[]{\label{c4}\includegraphics[height=2.41in]{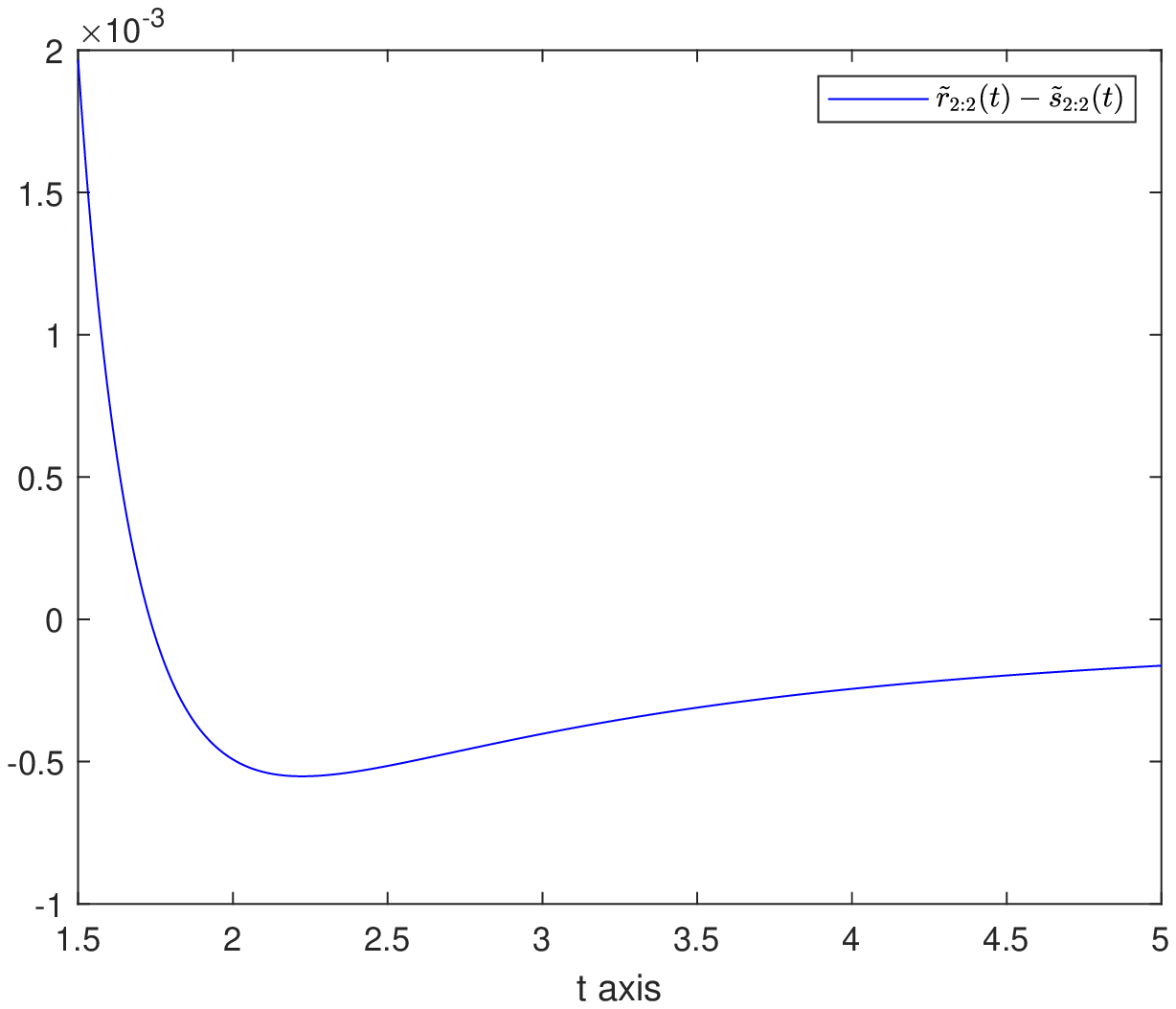}}
            \subfigure[]{\label{c3}\includegraphics[height=2.41in]{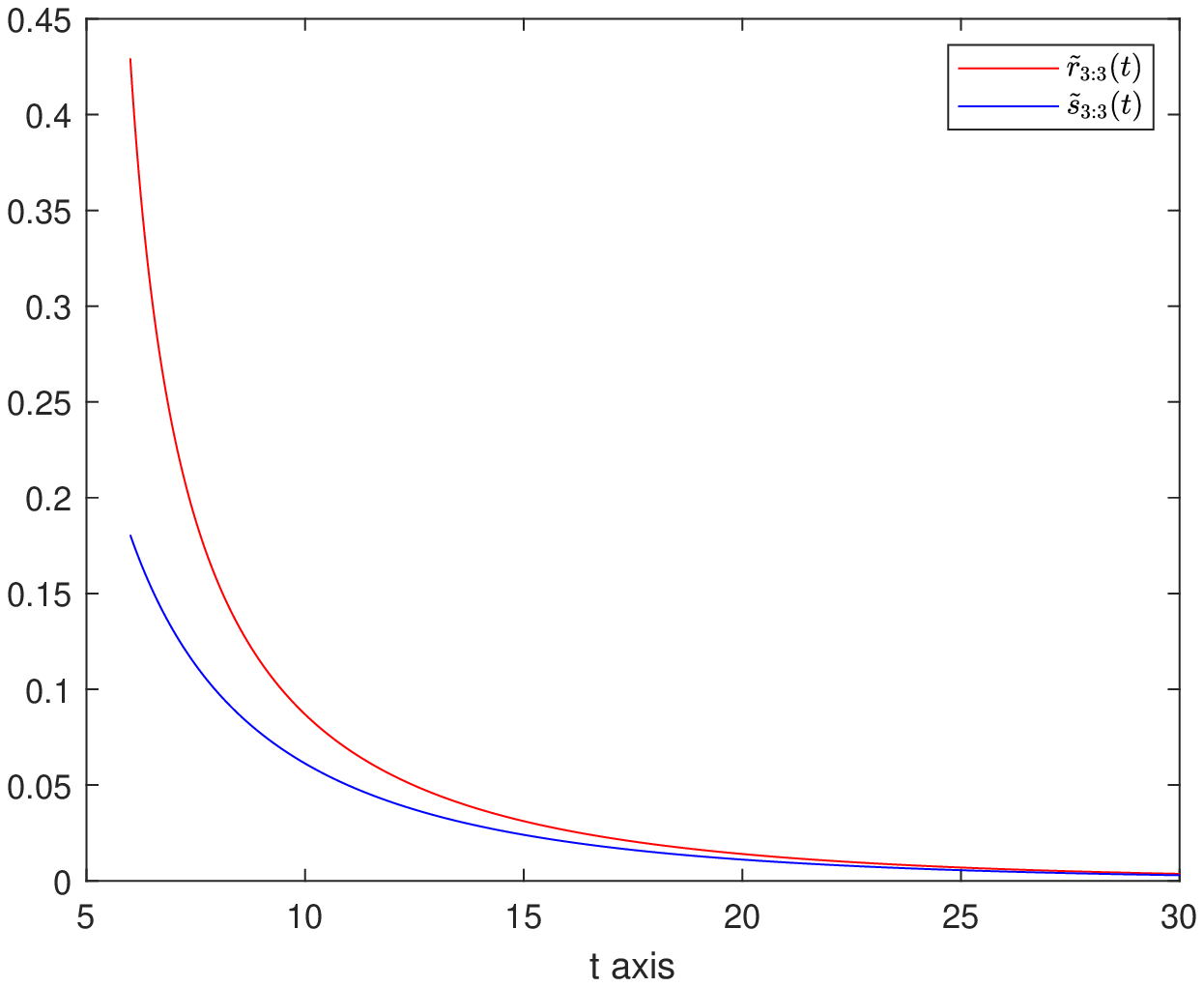}}
            \caption{(a) Graph of $\tilde{r}_{{2:2}}(t)-\tilde{s}_{{2:2}}(t)$ for  Counterexample \ref{ce3.3}. (b) Graphs of $\tilde{r}_{3:3}(t)$ and $\tilde{s}_{3:3}(t)$ for  Example \ref{ex3.4}.}
        \end{center}
    \end{figure}
\end{counterexample}

In analogy to Theorem \ref{th4.2}, the result in
Theorem \ref{th4.4} can also be generalized for
arbitrary $n\ge 3$. Here, we use a single
$T$-transform matrix to get
$(\bm{\mu},1/\bm{\delta};n)$,
$(\bm{\mu},\bm{\psi}(\bm{q});n)$ and
$(1/\bm{\delta},\bm{\psi}(\bm{q});n)$ from
$(\bm{\lambda},1/\bm{\theta};n)$,
$(\bm{\lambda},\bm{\psi}(\bm{p});n)$ and
$(1/\bm{\theta},\bm{\psi}(\bm{p});n)$,
respectively. That is,
$(\bm{\lambda},1/\bm{\theta};n)\gg
(\bm{\mu},1/\bm{\delta};n)$,
$(\bm{\lambda},\bm{\psi}(\bm{p});n)\gg
(\bm{\mu},\bm{\psi}(\bm{q});n)$ and
$(1/\bm{\theta},\bm{\psi}(\bm{p});n)\gg
(1/\bm{\delta},\bm{\psi}(\bm{q});n)$ hold.

\begin{theorem}\label{th4.5}
 Let $(A1)$ and $ (A2)$ hold, and $\bm{\alpha}=\bm{\beta}=\bm{1}_n$. Let $T_w$ be a $T$-transform matrix.
\begin{itemize}
\item[$(i)$] If $\bm{\theta}=\bm{\delta}~(=\theta\bm{1}_n)$,
$\bm{p}=\bm{q}~(=p\bm{1}_n)$ and
$(\bm{\lambda},1/\bm{\theta};n)\in Q_{n}$ hold,
then we have
$(\bm{\mu},1/\bm{\delta};n)=(\bm{\lambda},1/\bm{\theta};n)T_{w}\Rightarrow
U_{n:n}\ge_{rh}V_{n:n}$, provided $(C2)$, $(C3)$,
$(C4)$ and $(C5)$ are satisfied;

\item[$(ii)$] Suppose   $\psi$  satisfies $(C9)$.
Then, for
$\bm{\theta}=\bm{\delta}~(=\theta\bm{1}_n)$ and
$(\bm{\lambda},\bm{\psi}(\bm{p});n)\in M_n$, we
have
$(\bm{\mu},\bm{\psi}(\bm{q});n)=(\bm{\lambda},\bm{\psi}(\bm{p});n)T_{w}\Rightarrow
U_{n:n}\ge_{rh}V_{n:n}$, provided $(C1),~(C8)$
hold;

\item[$(iii)$] Suppose  $\psi$ satisfies $(C9)$.  Then, for
$\bm{\lambda}=\bm{\mu}~(=\lambda\bm{1}_n)$ and
$(1/\bm{\theta},\bm{\psi}(\bm{p});n)\in Q_{n}$,
we have $(1/\bm{\delta},\bm{\psi}(\bm{q});n)
=(1/\bm{\theta},\bm{\psi}(\bm{p});n)T_{w}\Rightarrow
U_{n:n}\ge_{rh}V_{n:n}$, provided  $(C2),$ $(C5)$
hold.

\end{itemize}
\end{theorem}
\begin{proof} Using Lemma \ref{lem2.3}, the proof of the theorem
follows similar to Theorem \ref{th4.2}. Hence, it
is omitted for the sake of brevity.
\end{proof}
Similar to Corollary \ref{th4.2a}, the above
result in Theorem \ref{th4.5} can be extended
from a $T$-transform matrix to a finite number of
$T$-transform matrices having same structure.
This is presented in the following corollary. In
this sequel, we denote $T^{*}_{w}=T_{w_1}\cdots
T_{w_{k}},$ where $k$ is finite.

\begin{corollary}\label{th4.5a}
    Under the assumptions in $(A1)$ and $(A2)$, let $\bm{\alpha}=\bm{\beta}=\bm{1}_n.$
    \begin{itemize}
        \item[$(i)$] Then, for $\bm{\theta}=\bm{\delta}~(=\theta\bm{1}_n)$,
        $\bm{p}=\bm{q}~(=p\bm{1}_n)$ and $(\bm{\lambda},1/\bm{\theta};n)\in Q_{n}$, we get $(\bm{\mu},1/\bm{\delta};n)=(\bm{\lambda},1/\bm{\theta};n)T^{*}_{w}\Rightarrow U_{n:n}\ge_{rh}V_{n:n}$, provided $(C2)$, $(C3)$, $(C4)$ and $(C5)$ hold;

        \item[$(ii)$] Suppose $\psi$ satisfies $(C9)$. Let $(C1)$ and $(C8)$ hold.
        Then, for $\bm{\theta}=\bm{\delta}~(=\theta\bm{1}_n)$ and
        $(\bm{\lambda},\bm{\psi}(\bm{p});n)\in M_n$, we get $(\bm{\mu},\bm{\psi}(\bm{q});n)=(\bm{\lambda},\bm{\psi}(\bm{p});n)$
        $T^{*}_{w}
        \Rightarrow U_{n:n}\ge_{rh}V_{n:n}$;

        \item[$(iii)$] Suppose  $\psi$ satisfies $(C9)$.
        Let $(C2)$ and $(C5)$ be satisfied. Further, assume $\bm{\lambda}=\bm{\mu}~(=\lambda\bm{1}_n)$ and $(1/\bm{\theta},\bm{\psi}(\bm{p});n)\in Q_{n}$. Then, $(1/\bm{\delta},\bm{\psi}(\bm{q});n)=(1/\bm{\theta},\bm{\psi}(\bm{p});n)T^{*}_{w}\Rightarrow U_{n:n}\ge_{rh}V_{n:n}$.

    \end{itemize}
\end{corollary}

Next result states that the results in Corollary
\ref{th4.2} also hold, when we have $T$-transform
matrices with different structures instead of the
same structure. The proof is similar to that of
Theorem \ref{th4.3}, and therefore, it is
omitted.
\begin{theorem}\label{th4.6}
Under $(A1)$ and $(A2)$, let
$\bm{\alpha}=\bm{\beta}=\bm{1}_n.$ Further, let
$T_{w_{1}},\cdots,T_{w_{j}}$ be the $T$-transform
matrices with different structures, for
$j=1,\cdots,k-1,$ where $k\geq 2$.

\begin{itemize}
    \item[$(i)$] Then, for
      $\bm{\theta}=\bm{\delta}~(=\theta\bm{1}_n)$,
      $\bm{p}=\bm{q}~(=p\bm{1}_n)$ and $(\bm{\lambda},1/\bm{\theta};n)\in Q_{n}$, we get $(\bm{\mu},1/\bm{\delta};n)=(\bm{\lambda},1/\bm{\theta};n)T_{w_{1}}\cdots T_{w_{k}}$
      $\Rightarrow U_{n:n}\ge_{rh}V_{n:n}$, provided $(C2)$, $(C3)$, $(C4)$ and $(C5)$ hold;

   \item[$(ii)$] Suppose $\psi$ satisfies $(C9)$. Let $(C1)$ and $(C8)$ hold.
   Then, for $\bm{\theta}=\bm{\delta}~(=\theta\bm{1}_n)$ and $(\bm{\lambda},\bm{\psi}(\bm{p});n)\in M_n$, we get $(\bm{\mu},\bm{\psi}(\bm{q});n)=(\bm{\lambda},\bm{\psi}(\bm{p});n)$
   $T_{w_{1}}\cdots$
   $ T_{w_{k}}$ $\Rightarrow U_{n:n}\ge_{rh}V_{n:n}$;

\item[$(iii)$] Suppose $\psi$ satisfies $(C9)$. For
$\bm{\lambda}=\bm{\mu}~(=\lambda\bm{1}_n)$ and
$(1/\bm{\theta},\bm{\psi}(\bm{p});n)\in Q_{n}$,
we have
$(1/\bm{\delta},\bm{\psi}(\bm{q});n)=(1/\bm{\theta},\bm{\psi}(\bm{p});n)T_{w_{1}}
\cdots T_{w_{k}}\Rightarrow
U_{n:n}\ge_{rh}V_{n:n}$, provided $(C2)$ and
$(C5)$ hold.
\end{itemize}
\end{theorem}

\section{Vector majorization\setcounter{equation}{0}}
This section is devoted to derive sufficient
conditions for the comparison results between the
largest claim amounts of two heterogeneous
portfolios of risks in the sense of the usual
stochastic and the reversed hazard rate
orderings. First, we present the results for the
usual stochastic ordering. The following theorem
shows that under some conditions, the weakly
supermajorized vector of shape parameters leads
to larger largest claim amount in the sense of
the usual stochastic order.  Here, we assume that
both the location and scale parameters are same
and fixed.

\begin{theorem}\label{th3.1a}
    Let $(A1)$ and $(A2)$ hold. Also,
for  $i=1,\cdots,n$, assume
$\theta_{i}=\delta_{i}=\theta$, $p_i=q_i$,
$\lambda_{i}=\mu_{i}=\lambda$ and
$\bm{\alpha},\bm{\beta}\in
\mathcal{E_+}(\mathcal{D_+}), $ $\bm{p}\in
\mathcal{D_+}(\mathcal{E_+})$.
    Then, $\bm{\alpha}\succeq^{w}\bm{\beta}$ implies $  U_{n:n}\le_{st}V_{n:n}.$

\end{theorem}
\begin{proof}
First, we consider the case
$\bm{\alpha},\bm{\beta}\in \mathcal{E_+}$ and
$\bm{p}\in \mathcal{D_+}$. Under the given set
up, the distribution function of $U_{n:n}$ is
obtained as
    \begin{eqnarray}\label{eq3.13}
    F_{n:n}(t)
    =\prod_{i=1}^{n}\left[1-p_{i}\left[1-{F}^{\alpha_i}\left(\frac{t-\lambda}{\theta}\right)\right]\right].
    \end{eqnarray}
     After taking the partial derivative of $F_{n:n}(t)$ with respect to $\alpha_{i}$, we
     get
    \begin{eqnarray}\label{eq3.14}
    \frac{\partial F_{n:n}(t)}{\partial \alpha_{i}}
    =G_iF_{n:n}(t),
    \end{eqnarray}
    where $G_i=(p_{i}\ln \left[F(\frac{t-\lambda}{\theta})\right]{F}^{\alpha_i}
    \left(\frac{t-\lambda}{\theta}\right))/(1-p_{i}\left[1-{F}^{\alpha_i}
    \left(\frac{t-\lambda}{\theta}\right)\right]).$
    For $i\leq j,$ we have $p_{i}\ge p_{j},~\alpha_{i}\le\alpha_{j}$. Thus, by Lemma \ref{lem2.7},
    we obtain $G_{i}\le G_{j}.$ Hence, the difference
    $\left[\frac{\partial F_{n:n}(t)}{\partial \alpha_{i}}-\frac{\partial F_{n:n}(t)}
    {\partial \alpha_{j}}\right]$ can be shown to be at most zero. This implies that
     $F_{n:n}(t)$ is Schur-convex with respect to $\bm{\alpha}\in \mathcal{E_+}$. Also, it is decreasing with respect to $\bm{\alpha}\in \mathcal{E_+}$.
     Thus, the rest of the proof is completed by Theorem $A.8$ of
\cite{Marshall2011}. The proof is similar when
$\bm{\alpha},\bm{\beta}\in
\mathcal{D_+}$ and $ \bm{p}\in
\mathcal{E_+}$.
    \end{proof}
In the following theorem, we consider that the
location, scale and shape parameters are same,
but vector valued.
\begin{theorem}\label{th3.2}
Suppose $\psi$ be a differentiable function. Let
$(A1)$ and $(A2)$ hold. Further, take
$\alpha_{i}=\beta_{i}$, $\theta_{i}=\delta_{i}$
and $\lambda_{i}=\mu_{i}$, for $i=1,\cdots,n.$
 If $\psi(.)$ satisfies $(C9)$, then
        $\psi(\bm{p})\succeq_{w}\psi(\bm{q})\Rightarrow U_{n:n}\ge_{st}V_{n:n}$, provided $\bm{\theta},\bm{\lambda},\bm{\alpha}, \bm{p}, \bm{q}\in \mathcal{E_+}(\mathcal{D_+})$.

\end{theorem}
\begin{proof}
We take $\bm{\theta},\bm{\lambda},\bm{\alpha},
\bm{p}, \bm{q}\in \mathcal{E_+}.$ The proof of
the other case is similar. Note that the
distribution function of $U_{n:n}$ is
    \begin{eqnarray}\label{eq3.1a}
    F_{n:n}(t)
    =\prod_{i=1}^{n}\left[1-\psi^{-1}(w_{i})\left[1-{F}^{\alpha_i}\left(\frac{t-\lambda_i}{\theta_i}\right)\right]\right],
    \end{eqnarray}
    where $p_i=\psi^{-1}(w_i)$, for $i=1,\cdots,n.$ Differentiating $F_{n:n}(t)$
with respect to $w_{i}$ partially, we obtain
    \begin{eqnarray}\label{eq3.5.}
    \frac{\partial F_{n:n}(t)}{\partial w_{i}}
    =-\frac{\partial \psi^{-1}(w_{i})}{\partial w_{i}}C_iF_{n:n}(t),
    \end{eqnarray}
    where $C_i=([1-{F}^{\alpha_i}(\frac{t-\lambda_i}{\theta_i})])/(
    1-\psi^{-1}(w_{i})[1-{F}^{\alpha_i}(\frac{t-\lambda_i}{\theta_i})]).$
    Under the assumptions made, for $i\leq j,$ we have
$w_{i}\le
w_{j},~\lambda_{i}\leq\lambda_{j},~\theta_{i}\le\theta_{j},~\alpha_{i}\le\alpha_{j}$
and  $\psi(w)$ is increasing, convex. Therefore,
$\psi^{-1}(.)$ is also increasing and  convex.
This implies $\psi^{-1}(w_{i})\le
\psi^{-1}(w_{j})$. Since $\psi^{-1}(.)$ is
convex, we can write $\frac{\partial
\psi^{-1}(w_{i})}{\partial w_{i}}\le
\frac{\partial \psi^{-1}(w_{j})}{\partial
w_{j}}$. Further, by Lemma \ref{lem2.6}, we have
$C_{i}\le C_{j}.$ Combining these two
inequalities, it can be checked that Equation
(\ref{eq3.5.}) is negative, and the difference
$\left[\frac{\partial F_{n:n}(t)}{\partial
w_{i}}-\frac{\partial F_{n:n}(t)}{\partial
w_{j}}\right]$ is greater than or equals to zero.
Hence, by using Lemma $3.1$ of
\cite{kundu2016some}, $F_{n:n}(t)$ is
Schur-concave with respect to $\bm{w}\in
\mathcal{E_+}$. Also, $F_{n:n}(t)$ is decreasing.
Now, applying Theorem $A.8$ of
\cite{Marshall2011}, we get the required result.
\end{proof}

\begin{remark}
Theorem \ref{th3.2} demonstrates that more heterogeneity among transformed occurrence probabilities with respect to the weakly submajorization order provides better tail function of the largest claim amount.
\end{remark}

%

The following theorem shows that the usual stochastic ordering holds between the largest claim amounts, when the reciprocal of the scale parameters are associated with the $p$-majorization and reciprocal majorization orders. We assume that the shape parameters are less than or equal to $1$.
\begin{theorem}\label{th3.1}
    Suppose $(A1)$ and $(A2)$ hold. For  $i=1,\cdots,n$, let $\alpha_{i}=\beta_{i}=\alpha\leq 1$, $p_i=q_i$ and $\lambda_{i}=\mu_{i}$. Take $\bm{p},\bm{\theta},\bm{\delta},\bm{\lambda}\in \mathcal{E_+}(\mathcal{D_+})$.
    \begin{itemize}
    \item[$(i)$] If  $(C2)$ holds, then $1/\bm{\theta}\succeq^{p}1/\bm{\delta}$ implies $  U_{n:n}\ge_{st}V_{n:n};$
        \item[$(ii)$] If  $(C3)$ holds, then $1/\bm{\theta}\succeq^{rm}1/\bm{\delta}$ implies $ U_{n:n}\ge_{st}V_{n:n}.$
    \end{itemize}
\end{theorem}

\begin{proof}
    $(i)$ We prove the result when $\bm{p},\bm{\theta},\bm{\delta},\bm{\lambda}\in \mathcal{E_+}$. The proof for $\bm{p},\bm{\theta},\bm{\delta},\bm{\lambda}\in\mathcal{D_+}$ is analogous. Under the present set up, the distribution function of $U_{n:n}$ is
    \begin{eqnarray}\label{eq3.3}
    F_{n:n}(t)=\prod_{i=1}^{n}\left[1-p_{i}\left[1-{F}^{\alpha}\left({(t-\lambda_i)e^{s_i}}\right)\right]\right],
    \end{eqnarray}
    where $s_i=-\ln\theta_i$, for $i=1,\cdots,n.$  To prove the required result,
    we consider $\chi_1(e^{\boldsymbol{s}})=\prod_{i=1}^{n}\left[1-p_i\left[1-{F}^{\alpha}\left({(t-\lambda_i)e^{s_i}}\right)\right]\right].$
    Differentiating  $\chi_1(e^{\boldsymbol{s}})$ with respect to $s_{i}$ partially, we have
    \begin{eqnarray}\label{eq3.4}
    \frac{\partial \chi_1(e^{\boldsymbol{s}})}{\partial s_{i}}=\left[\frac{\alpha F^{(\alpha-1)}(x)[1-F(x)]}{1-F^{\alpha}(x)}\right]_{x=(t-\lambda_i)e^{s_i}}\left[xr(x)\right]_{x=(t-\lambda_i)e^{s_i}}A_i\chi_1(e^{\boldsymbol{s}}),
    \end{eqnarray} where $A_i=(\left[1-F^{\alpha}\left({(t-\lambda_i)e^{s_i}}\right)\right])/(1-p_i[1-F^{\alpha}\left({(t-\lambda_i)e^{s_i}}\right)])$.
    First, we consider the case $s_i\ge s_j,\lambda_i\leq\lambda_j,p_i\leq p_j$.
    Using the given assumptions and Lemma \ref{lem2.6}, the inequality $A_{i}\le A_{j}$ holds.
    Combining this with Lemma $3$ of \cite{balakrishnan2015stochastic} and the decreasing property of $xr(x)$, we obtain
    \begin{eqnarray}\label{eq3.5}
    \left[\frac{\partial \chi_1(e^{\boldsymbol{s}})}{\partial s_{i}}-\frac{\partial \chi_1(e^{\boldsymbol{s}})}{\partial s_{j}}\right]\leq 0.
    \end{eqnarray}
    From Lemma $2.1$ of \cite{khaledi2002dispersive}, (\ref{eq3.5}) implies that $\chi_1(e^{\boldsymbol{s}})$ is Schur-concave  with respect to $\bm{s}\in \mathcal{D_+}$. Further, $\chi_1(e^{\boldsymbol{s}})$ is increasing  with respect to $\bm{s}$. Thus, by Theorem $A.8$ of \cite{Marshall2011}, the rest of the proof is completed. \\
    $(ii)$
    The cumulative distribution function of $U_{n:n}$ is
    \begin{eqnarray}\label{eq3.6}
    F_{n:n}(t)
    =\prod_{i=1}^{n}\left[1-p_i\left[1-{F}^{\alpha}\left(\frac{t-\lambda_i}{\theta_i}\right)\right]\right].
    \end{eqnarray}
    Denote $\chi_2(1/\boldsymbol{\theta})=\prod_{i=1}^{n}[1-p_i[1-{F}^{\alpha}
    (\frac{t-\lambda_i}{\theta_i})]].$
    Differentiating, $\chi_2(1/\boldsymbol{\theta})$ with
    respect to $\theta_{i}$ partially, we get
    \begin{eqnarray}\label{eq3.6a}
    \frac{\partial \chi_2(1/\boldsymbol{\theta})}{\partial \theta_{i}}=-\left[\frac{\alpha F^{(\alpha-1)}(x)[1-F(x)]}{1-F^{\alpha}(x)}\right]_{x=\left(\frac{t-\lambda_i}{\theta_i}\right)}
    \frac{\left[x^2r(x)\right]_{x=\left(\frac{t-\lambda_i}{\theta_i}\right)}}
    {(t-\lambda_i)}B_i\chi_2(1/\boldsymbol{\theta}),
    \end{eqnarray}
    where $B_i=([1-F^{\alpha}(\frac{t-\lambda_i}{\theta_i})])/(1-p_i[1-F^{\alpha}(\frac{t-\lambda_i}{\theta_i})
    ])$.  Now, similar to the arguments used in the proof of Part $(i)$ and by Lemma $1$ of \cite{hazra2017stochastic}, it can be established that $\chi_{2}(1/\boldsymbol{\theta})$ is decreasing and Schur-concave with respect to $\bm{\theta}\in\mathcal{E_+}(\mathcal{D_+})$. Hence, the proof is completed.
\end{proof}

In the next, we prove ordering result between two largest claim amounts in the sense of the usual stochastic order, when the vectors of the location parameters are connected with the weak submajorization order. We consider that the scale and risk parameters are same, but vector valued. The proof is omitted since it follows similarly to that of Theorem \ref{th3.1}$(i)$.
\begin{theorem}\label{th3.2.}
    Let $(A1)$, $(A2)$ and $(C2)$ hold. If $\bm{\theta},\bm{\lambda},\bm{\mu}, \bm{p}\in \mathcal{D_+}(\mathcal{E_+})$,  then we have
    $\bm{\lambda}\succeq_{w}\bm{\mu}\Rightarrow U_{n:n}\ge_{st}V_{n:n},$ provided  $\bm{\alpha}=\bm{\beta}=\alpha
    \bm{1}_n~(\alpha\leq 1)$, $\bm{\theta}=\bm{\delta}$ and $\bm{p}=\bm{q}$.
\end{theorem}

Below, we obtain two different sets of sufficient conditions for the existence of the usual stochastic order between the largest claim amounts arising from heterogeneous portfolios of risks for the case of different location,  scale and  risk  parameters.

\begin{theorem}\label{th3.12}
    Let a function $\psi$ be differentiable. Suppose $(A1)$ and $(A2)$ hold. Further, let $\bm{\alpha}=\bm{\beta}=\alpha
    \bm{1}_n~(\alpha\leq 1)$ and $\bm{\theta},\bm{\lambda},\bm{\mu},\bm{\delta}, \bm{p}, \bm{q}\in \mathcal{E_+}(\mathcal{D_+})$.
    \begin{itemize}
    \item[$(i)$] If  $(C2)$ and $(C9)$ hold, then $1/\bm{\theta}\succeq^{p}1/\bm{\delta}$, $\psi(\bm{p})\succeq_{w}\psi(\bm{q})$ and  $\bm{\lambda}\succeq_{w}\bm{\mu}$ imply $  U_{n:n}\ge_{st}V_{n:n};$
        \item[$(ii)$] If  $(C2)$, $(C3)$ and $(C9)$ hold, then $1/\bm{\theta}\succeq^{rm}1/\bm{\delta},$ $\psi(\bm{p})\succeq_{w}\psi(\bm{q})$ and  $\bm{\lambda}\succeq_{w}\bm{\mu}$ imply $ U_{n:n}\ge_{st}V_{n:n}.$
    \end{itemize}
\end{theorem}
\begin{proof}
The proof of the first (second) part of the theorem follows from Theorem \ref{th3.2}, Theorem \ref{th3.1}$(i)$ ($(ii)$) and Theorem \ref{th3.2.}. So, it is omitted.
\end{proof}

The following theorem states that the $k$th largest claim amounts can be comparable with respect to the usual stochastic order.

\begin{theorem}\label{th3.9}
    Let $(A1)$ and $(A2)$ hold, and
    $\psi$ be a
    differentiable function. Also, let  $\lambda_i=\mu_i=
    \lambda$, $\theta_i=\delta_i=\theta$, for $i=1,\cdots,n$.
    Then, $\boldsymbol {\alpha}\succeq^m\boldsymbol\beta
        \Rightarrow U_{k:n}\leq_{st} V_{k:n}$, provided $\psi(\bm{p})=\psi(\bm{q})=v$.
        \end{theorem}
\begin{proof}
    Consider the function
    $F(\alpha_i,t)
    =1-{\psi}^{-1}(v)\left[1-{F}^{\alpha_i}\left(\frac{t-\lambda}{\theta}\right)\right]$, which can be shown to be increasing and log-convex with respect to $\alpha_i$, for $i=1,\cdots,n$. Thus, the result follows from Theorem $3.5$ of \cite{pledger1971comparison}.
\end{proof}

In this part, we study the conditions under which the reversed hazard rate order holds between the largest claim amounts from two heterogeneous insurance portfolios of risks.
\begin{theorem}\label{th3.3.}
    Assume that  $(A1)$ and $(A2)$ hold. Further, let $\bm{\alpha}=\bm{\beta}=\bm{1}_n$, $\bm{p}=\bm{q}$,
    $\bm{\theta}=\bm{\delta}$ and $\bm{\lambda},\bm{\theta},\bm{\mu}, \bm{p}\in\mathcal{E}_{+}(\mathcal{D}_{+})$.
    Then, $\bm{\lambda}\succeq_{w}\bm{\mu}\Rightarrow
    U_{n:n}\ge_{rh}V_{n:n}$ if  $(C2),$ $(C3)$ and $(C4)$ are satisfied.
    \end{theorem}
\begin{proof}
    Under the set up, the reversed hazard rate function of $U_{n:n}$ is
    \begin{eqnarray}\label{eq3.8}
    \tilde{r}_{n:n}(t)=\sum_{i=1}^{n}\frac{1}{\theta_i}r
    \left(\frac{t-\lambda_i}{\theta_i}\right)\left[
    \frac{p_i\left[1-F\left(\frac{t-\lambda_i}{\theta_i}\right)\right]}
    {1-p_i\left[1-F\left(\frac{t-\lambda_i}
        {\theta_i}\right)\right]}\right].
    \end{eqnarray}
    The partial derivative of $\tilde{r}_{n:n}(t)$ given by (\ref{eq3.8}) with respect to
    $\lambda_{i}$, is given by
    \begin{eqnarray}\label{eq3.9}
    \frac{\partial \tilde{r}_{n:n}(t)}{\partial
         \lambda_{i}}=&-\frac{\left[x^2r(x)\right]_{x=\left(\frac{t-\lambda_i}{\theta_i}\right)}}{(t-\lambda_i)^2}\left[\frac{r'(x)}{r(x)}\right]_{x=\left(\frac{t-\lambda_i}{\theta_i}\right)}\left[
    \frac{p_i\left[1-F\left(\frac{t-\lambda_i}{\theta_i}\right)\right]}
    {1-p_i\left[1-F\left(\frac{t-\lambda_i}
        {\theta_i}\right)\right]}\right]\\
    &+\frac{1}{(t-\lambda_i)^2}\left[xr(x)\right]^2_{x=\left(\frac{t-\lambda_i}{\theta_i}\right)}\left[
    \frac{p_i\left[1-F\left(\frac{t-\lambda_i}{\theta_i}\right)\right]}
    {\left[1-p_i\left[1-F\left(\frac{t-\lambda_i}
        {\theta_i}\right)\right]\right]^2}\right]\nonumber.
    \end{eqnarray}
    It is easy to see that $\tilde{r}_{n:n}(t)$ is increasing
    with respect to $\lambda_{i}$, $i=1,\cdots,n$, since $r(x)$ is decreasing.
    For any $1\le i\le j\le n$, we have
    $\frac{t-\lambda_i}{\theta_i}\ge(\le)\frac{t-\lambda_j}{\theta_j}$. Applying the property of $\frac{r'(x)}{r(x)}~(\le 0)$ and $xr(x)$, we obtain $$-\left[\frac{r'(x)}{r(x)}\right]_{x=\frac{t-\lambda_i}{\theta_i}}\le(\ge)-\left[\frac{r'(x)}{r(x)}\right]_{x=\frac{t-\lambda_j}{\theta_j}} \text{ and } [{xr(x)}]_{x=\frac{t-\lambda_i}{\theta_i}}\le(\ge)[{xr(x)}]_{x=\frac{t-\lambda_j}{\theta_j}}.$$
    With these results, Lemma \ref{lem2.6} yields that
    $[\frac{\partial
        \tilde{r}_{n:n}(t)}{\partial
        \lambda_{i}}-\frac{\partial
        \tilde{r}_{n:n}(t)}{\partial \lambda_{j}}]$ is less than or equals (greater than or equals) to zero.
    Thus, $\tilde{r}_{n:n}(t)$ is Schur-convex with
    respect to
    $\bm{\lambda}\in\mathcal{E}_{+}(\mathcal{D}_{+})$
by using Lemma $3.3$ $(3.1)$ of
    \cite{kundu2016some}. The rest of the proof is completed by Theorem $A.8$ of
    \cite{Marshall2011}.
\end{proof}

In the next theorem, we  consider that the shape parameter vectors are equal to $\bm{1}_n$ and the location parameter vectors are same.
\begin{theorem}\label{th3.3}
    Suppose $(A1)$ and $(A2)$  hold. Further, we assume $\bm{\alpha}=\bm{\beta}=\bm{1}_n$, $\bm{p}=\bm{q}$, $\bm{\lambda}=\bm{\mu}$ and $\bm{\lambda},\bm{\theta},\bm{\delta}, \bm{p}\in\mathcal{E}_{+}(\mathcal{D}_{+})$.
    \begin{itemize}

        \item[$(i)$] If $(C2)$  and $(C5)$ hold, then $1/\bm{\theta}\succeq^{w}1/\bm{\delta}\Rightarrow
        U_{n:n}\ge_{rh}V_{n:n}$;
        \item[$(ii)$] If $(C2),$ $(C6)$ and $(C7)$ hold, then $1/\bm{\theta}\succeq^{rm}1/\bm{\delta}\Rightarrow
        U_{n:n}\ge_{rh}V_{n:n}$.
    \end{itemize}
\end{theorem}
\begin{proof}

    $(i)$ Under the given set up, we have
    \begin{eqnarray}\label{eq3.7.}
    \tilde{r}_{n:n}(t)=\sum_{i=1}^{n}{m_{i}}r
    \left((t-\lambda_i){m_i}\right)\left[
    \frac{p_i\left[1-F\left({(t-\lambda_i)}{m_i}\right)\right]}
    {1-p_i\left[1-F\left({(t-\lambda_i)}
        {m_i}\right)\right]}\right],
    \end{eqnarray} where $m_i=1/ \theta_i$, for $i=1,\cdots,n.$
    Taking derivative of  (\ref{eq3.7.}) with respect to
    $m_{i}$ partially, we obtain
    \begin{eqnarray}\label{eq3.8.}
    \frac{\partial \tilde{r}_{n:n}(t)}{\partial
        m_{i}}=&\frac{\partial}{\partial x}\left[{xr(x)}\right]_{x=\left((t-\lambda_i){m_i}\right)}\left[
    \frac{p_i\left[1-F\left((t-\lambda_i){m_i}\right)\right]}
    {1-p_i\left[1-F\left((t-\lambda_i){m_i}\right)\right]}\right]\\
    &-\left[xr^2(x)\right]_{x=\left((t-\lambda_i){m_i}\right)}\left[
    \frac{p_i\left[1-F\left((t-\lambda_i){m_i}\right)\right]}
    {\left[1-p_i\left[1-F\left((t-\lambda_i){m_i}\right)\right]\right]^2}\right]\nonumber.
    \end{eqnarray}
    From \eqref{eq3.8.}, it is clear that
    $\tilde{r}_{n:n}(t)$ is decreasing
    with respect to $m_{i}$, for $i=1,\cdots,n$.
    Now, let us take $1\le i\le j\le n$. Then,
    $(t-\lambda_i){m_i}\ge(\le)(t-\lambda_j){m_j}$. Moreover, $xr(x)$ is decreasing and convex. Hence, we get  $\frac{\partial}{\partial x}\left[{xr(x)}\right]_{x=\left((t-\lambda_i){m_i}\right)}\ge(\le)\frac{\partial}{\partial x}\left[{xr(x)}\right]_{x=\left((t-\lambda_j){m_j}\right)}$ and $[{xr(x)}]_{x=({t-\lambda_i}){m_i}}\le(\ge)[{xr(x)}]_{x=({t-\lambda_j}){m_j}}$. Utilizing
    Lemma \ref{lem2.6}, we can show  that for any $i\le j$, $\frac{\partial
        \tilde{r}_{n:n}(t)}{\partial
        m_{i}}-\frac{\partial
        \tilde{r}_{n:n}(t)}{\partial m_{j}}\ge(\le)0$. Thus,
 $\tilde{r}_{n:n}(t)$ is Schur-convex with respect to    $\bm{m}\in\mathcal{D}_{+}(\mathcal{E}_{+})$. The rest of the proof follows from Lemma $2.1$ of \cite{khaledi2002dispersive} and Theorem $A.8$ of \cite{Marshall2011}. The second part of the theorem can be proved in a similar manner by using Lemma $1$ of \cite{hazra2017stochastic}. Thus, it is omitted.
\end{proof}

%

\begin{theorem}\label{th3.4}
    Let  $(A1)$, $(A2)$ and $(C2)$ hold. Then,
    \begin{itemize}
        \item [$(i)$]
        $\{\bm{\alpha}=\bm{\beta}= \bm{1}_n$,
        $\bm{\theta}\geq\bm{\delta}$,
        $\bm{\lambda}\geq\bm{\mu},~ \bm{p}\geq\bm{q}\} \Rightarrow
        U_{n:n}\geq_{rh}V_{n:n}$;
        \item [$(ii)$]
        $\{\bm{\alpha}=\bm{\beta}= \alpha<1$,
        $\bm{\theta}\geq\bm{\delta}$,
        $\bm{\lambda}\geq\bm{\mu},~ \bm{p}\geq\bm{q}\} \Rightarrow
        U_{n:n}\geq_{rh}V_{n:n}$.

    \end{itemize}
\end{theorem}
\begin{proof}
To prove the first part, it suffices to show that
    \begin{equation}\label{eq3.10}
    \sum_{i=1}^{n}\frac{p_i}{\theta_i}r\left(\frac{t-\lambda_i}{\theta_i}\right)s_i\geq
    \sum_{i=1}^{n}\frac{q_i}{\delta_i}r\left(\frac{t-\mu_i}{\delta_i}\right)t_i,
    \end{equation}
    where
    $s_i=(p_i[1-F(\frac{t-\lambda_i}{\theta_i})])/(1-p_i[1-F(\frac{t-\lambda_i}
        {\theta_i})]$ and
    $t_i=(q_i[1-F(\frac{t-\mu_i}{\delta_i})])
   (1-q_i[1-F(\frac{t-\mu_i}
        {\delta_i})]).$
    Note that the inequality given by
    $\eqref{eq3.10}$ holds if
    \begin{eqnarray}\label{eq3.11}
    \frac{1}{\delta_i}{r}\left(\frac{t-\mu_i}{\delta_i}\right)
    \le\frac{1}{\theta_i}{r}\left(\frac{t-\lambda_i}{\theta_i}\right)
    \end{eqnarray}
    and $s_i\geq t_i$, for all $i=1,\cdots,n$. Thus, the proof is completed by the given assumptions and Lemma \ref{lem2.6}. The second part of the theorem can be proved by  Lemmas
    $3(i),3(ii)$ of \cite{balakrishnan2015stochastic}.
\end{proof}

Next theorem shows that $V_{n:n}$ is dominated by $U_{n:n}$ with respect to the reversed hazard rate order under some conditions.  Here, we take that the location parameter and scale parameter vectors are equal. The shape parameters are taken fixed and equal to $1$.
\begin{theorem}\label{th3.5}
Suppose $\psi:(0,1)\rightarrow(0,\infty)$ be a differentiable function satisfying $(C9)$.
    Let $(A1)$, $(A2)$ and $(C2)$  hold. Again, $\bm{\theta}=\bm{\delta},$ $\bm{\lambda}=\bm{\mu}$, $\bm{\alpha}=\bm{\beta}=\alpha\bm{1}_n~(\leq 1)$ and $\bm{\lambda},\bm{\theta},\bm{p},\bm{q}\in\mathcal{E_+}(\mathcal{D_+})$.
    Then, $\psi(\bm{p})\succeq_{w}\psi(\bm{q}) \Rightarrow
    U_{n:n}\geq_{rh}V_{n:n}$.
\end{theorem}
\begin{proof}
The reversed hazard rate of $U_{n:n}$ is
    \begin{eqnarray}\label{eq3.11.}
    \tilde{r}_{n:n}(t)=\sum_{i=1}^{n}\frac{1}{\theta_i}r
    \left(\frac{t-\lambda_i}{\theta_i}\right)D_iG_i,
    \end{eqnarray}
    where $D_i=
    ([1-F^{\alpha}(\frac{t-\lambda_i}{\theta_i})])/
    (1-\psi^{-1}(w_{i})[1-F^{\alpha}(\frac{t-\lambda_i}
        {\theta_i})])$ and $G_i=(\alpha F^{\alpha-1}(\frac{t-\lambda_i}{\theta_i})[1-F(\frac{t-\lambda_i}{\theta_i})])/([1-F^{\alpha}(\frac{t-\lambda_i}{\theta_i})])$.
    On differentiating (\ref{eq3.11.}) with respect to
    $w_{i}$ partially, we get
    \begin{eqnarray}\label{eq3.12}
    \frac{\partial \tilde{r}_{n:n}(t)}{\partial
        w_{i}}=\frac{\partial\psi^{-1}(w_{i})}{\partial w_{i}}\frac{D_iG_i}{1-\psi^{-1}(w_{i})\left[1-F^{\alpha}\left(\frac{t-\lambda_i}
        {\theta_i}\right)\right]}\frac{1}{\theta_i}r
    \left(\frac{t-\lambda_i}{\theta_i}\right).
    \end{eqnarray}
    Since $\psi(w)$ is increasing,  $\psi^{-1}(w)$ is also increasing. Therefore,  $\tilde{r}_{n:n}(t)$ is increasing with respect to $w_{i}$, $i=1,\cdots,n$.
    Now, under the assumptions made, we have
    $\frac{t-\lambda_i}{\theta_i}\ge(\le)\frac{t-\lambda_j}{\theta_j}$. Using $(C2)$, we can write  $[{xr(x)}]_{x=\frac{t-\lambda_i}{\theta_i}}\le(\ge)[{xr(x)}]_{x=\frac{t-\lambda_j}{\theta_j}}$. Further, by Lemma \ref{lem2.6}, and Lemma $3$ of \cite{balakrishnan2015stochastic}, we get
    $\frac{\partial
        \tilde{r}_{n:n}(t)}{\partial
        w_{i}}-\frac{\partial
        \tilde{r}_{n:n}(t)}{\partial w_{j}}\le(\ge)0$.
    Thus, $\tilde{r}_{n:n}(t)$ is Schur-convex with
    respect to
    $\bm{w}\in\mathcal{E}_{+}(\mathcal{D}_{+})$, and the desired result readily follows from Theorem $A.8$ of
    \cite{Marshall2011}.
\end{proof}
The following theorem is an extension of Theorem \ref{th3.12}. Proof of the first part of the theorem follows from Theorem \ref{th3.3.}, Theorem \ref{th3.3}$(i)$ and Theorem \ref{th3.5}. The second part follows from Theorem \ref{th3.3.}, Theorem \ref{th3.3}$(ii)$ and Theorem \ref{th3.5}.

\begin{theorem}\label{th3.13}
    Let $\psi$ be a differentiable function. Further, let  $(A1)$ and $(A2)$ hold. Also, assume $\bm{\alpha}=\bm{\beta}=\bm{1}_n$ and $\bm{\theta},\bm{\lambda},\bm{\mu},\bm{\delta}, \bm{p}, \bm{q}\in \mathcal{E_+}(\mathcal{D_+})$.
    \begin{itemize}                                                                 \item[$(i)$]
    Suppose $(C2), (C5)$ and $(C9)$ hold. Then, $1/\bm{\theta}\succeq^{w}1/\bm{\delta}$, $\psi(\bm{p})\succeq_{w}\psi(\bm{q})$ and  $\bm{\lambda}\succeq_{w}\bm{\mu}$ imply $  U_{n:n}\ge_{rh}V_{n:n};$
        \item[$(ii)$] Suppose $(C2), (C3), (C6), (C7)$ and $(C9)$ hold. Then, $1/\bm{\theta}\succeq^{rm}1/\bm{\delta},$ $\psi(\bm{p})\succeq_{w}\psi(\bm{q})$ and  $\bm{\lambda}\succeq_{w}\bm{\mu}$ imply $ U_{n:n}\ge_{rh}V_{n:n}.$
    \end{itemize}
\end{theorem}

\begin{remark}
On using Theorem $2.3$ of \cite{zardasht2015test},  the usual stochastic ordering implies the incomplete cumulative residual entropy ordering. Further, the reversed hazard rate ordering implies the usual stochastic ordering. Therefore, the results obtained in this paper also compare two largest claim amounts arising from two heterogeneous portfolios of risks in the sense of the incomplete cumulative residual entropy ordering.
\end{remark}

\section{Applications\setcounter{equation}{0}}
In this section, we consider two special probability models and show the applicability of the established results. We take generalized linear failure rate and Marshal-Olkin extended quasi Lindley distributions. For these distributions, we present few corollaries. However, one can easily obtain similar applications for other established results.
\subsection{Generalized linear failure rate distribution}
The distribution function of the generalized linear failure rate distribution is given by
\begin{equation}\label{fp}
F(x)=\left[1-e^{-(ax+\frac{b}{2}x^2)}\right]^d,~ x\geq 0,~ a,~b,~d>0.
\end{equation}
 One can easily check that for $d=0.5, a=1,b=0,$ the hazard rate function of  \eqref{fp} satisfies all the conditions given in $(C1)-(C4)$. Here, we consider generalized linear failure rate distribution as the baseline distribution function.
\begin{corollary}\label{cor6.1}
   Under the assumptions of Theorem \ref{th3.1} $(i)$,  $1/\bm{\theta}\succeq^{p}1/\bm{\delta}\Rightarrow   U_{n:n}\ge_{st}V_{n:n}.$
    \end{corollary}
\begin{corollary}\label{cor6.2}
    Let the assumptions of Theorem \ref{th3.3.} hold. Then, $\bm{\lambda}\succeq_{w}\bm{\mu}\Rightarrow   U_{n:n}\ge_{rh}V_{n:n}.$
\end{corollary}
Next, we consider an example to illustrate Corollary \ref{cor6.2}.
\begin{example}\label{ex3.4}
Consider $\{X_{1},X_{2},X_{3}\}$ and
    $\{Y_{1},Y_{2},Y_{3}\}$ to be the collections of
    independent random variables such that $X_{i}\sim
    F^{\alpha}(\frac{x-\lambda_{i}}{\theta_{i}})$ and $Y_{i}\sim
    F^{\alpha}(\frac{x-\mu_{i}}{\theta_{i}})$, for $i=1,2,3.$  Also, suppose that $\{J_{1},J_2,J_{3}\}$  is a set
    of independent Bernoulli random variables,
    independent of $X_{i}$'s  with $E(J_{i})=p_{i}$ and  $\{J^*_{1},J^*_2,J^*_{3}\}$  is another set
    of independent Bernoulli random variables,
    independent of $Y_{i}$'s  with $E(J^*_{i})=q_{i}$, for $i=1,2,3.$ Set $\bm{\lambda}=(1,2.5,5)$,
    $\bm{\mu}=(0.5,2,3),
    \bm{\theta}=(2,5,9)$, $\bm{p}=\bm{q}=(0.2,0.8,0.9),~\alpha=1.$ Clearly,
    $\bm{\lambda},\bm{\mu}\in\mathcal{E}_{+}$ and $\bm{\lambda}\succeq_{w}\bm{\mu}$.
 Also $(C2)-(C4)$ hold, for $d=0.5,a=1,b=0$.
    Thus, as an application of Corollary \ref{cor6.2}, we have
    $U_{3:3}\ge_{rh}V_{3:3}$. The graphs of $\tilde{r}_{3:3}(t)$ and $\tilde{s}_{3:3}(t)$ are given in Figure $2(b),$ that varifies Corollary \ref{cor6.3}.
\end{example}
\subsection{ Marshall-Olkin extended quasi Lindley distribution}
Let us consider the distribution function of the
Marshall-Olkin extended quasi Lindley
distribution as
\begin{equation}\label{hn}
    F(x)=\frac{1-\left(\frac{b+1+d x}{b+1}e^{-d x}\right)}{1-(1-a)\left(\frac{d(b+1+d x)}{b+1}e^{-d x}\right)}, ~x>0,~a,d>0,~ b>-1.
\end{equation}
For $a=0.1, b=-0.9$ and  $d=0.8,$ the hazard rate
function of the Marshall-Olkin extended quasi
Lindley distribution satisfies $(C1)-(C8)$.
Below, we present some corollaries, which  are
the direct consequences of Theorem \ref{th4.1}
$(i)$ and Theorem \ref{th3.13}$(i)$. We consider
the Marshall-Olkin extended quasi Lindley
distribution to be the baseline distribution
function.
\begin{corollary}\label{cor6.3}
Suppose the assumptions of Theorem \ref{th4.1} $(i)$ hold. Again, let $\psi(p)=p^2$. Then, we have  \begin{equation*}
    \left(\begin{smallmatrix} \psi(p_{1}) & \psi(p_{2}) \\ \lambda_{1} &
    \lambda_{2}
    \end{smallmatrix}\right) \gg \left(\begin{smallmatrix}
    \psi(q_{1}) & \psi(q_{2}) \\ \mu_{1} &
    \mu_{2}
    \end{smallmatrix}\right) \Rightarrow U_{2:2} \geq_{st} V_{2:2}.
    \end{equation*}
\end{corollary}
\begin{corollary}\label{cor6.4}
    Let us assume $\psi(p)=e^p$. If all the assumptions of Theorem  \ref{th3.13} $(i)$ hold, then $1/\bm{\theta}\succeq^{w}1/\bm{\delta}$, $\psi(\bm{p})\succeq_{w}\psi(\bm{q})$ and  $\bm{\lambda}\succeq_{w}\bm{\mu}$ imply $  U_{n:n}\ge_{rh}V_{n:n}.$

\end{corollary}
The example given below,  illustrates Corollary  \ref{cor6.3}.
\begin{example}\label{ex5.2}
Let $\{X_{1},X_{2}\}$ and
    $\{Y_{1},Y_{2}\}$ be two collections of
    independent random variables such that $X_{i}\sim
    F^{\alpha}(\frac{x-\lambda_{i}}{\theta})$ and $Y_{i}\sim
    F^{\alpha}(\frac{x-\mu_{i}}{\theta})$, for $i=1,2.$ Set the baseline distribution function as  Marshall-Olkin extended quasi Lindley distribution. Also, let $\{J_{1},J_2\}$  be a set
    of independent Bernoulli random variables,
    independent of $X_{i}$'s  with $E(J_{i})=p_{i}$ and  $\{J^*_{1},J^*_2\}$  be another set
    of independent Bernoulli random variables,
    independent of $Y_{i}$'s  with $E(J^*_{i})=q_{i}$, where $i=1,2.$ Consider $\bm{\lambda}=(5,6.1)$,
    $\bm{\mu}=(5.44,5.66)$,
    $\theta=0.01$, $\alpha=0.52,$ $\psi(\bm{p})=(0.2,0.5)$,
    $\psi(\bm{q})=(0.32,0.38)$. Here, $\psi(p)=p^2,$ which is increasing and convex with respect to $p$. Let $T_{0.6}=\left( \begin{smallmatrix} 0.6 & 0.4\\ 0.4 & 0.6\\
\end{smallmatrix} \right)$. Then, $\left(\begin{smallmatrix} \psi(p_{1}) & \psi(p_{2}) \\ \lambda_{1} &
    \lambda_{2}
\end{smallmatrix}\right) \gg \left(\begin{smallmatrix}
    \psi(q_{1}) & \psi(q_{2}) \\ \mu_{1} &
    \mu_{2}
\end{smallmatrix}\right)$.
    Thus, as an application of Corollary \ref{cor6.3},
    we have
    $ U_{2:2} \geq_{st} V_{2:2}$. The graph of $F_{2:2}(t)-G_{2:2}(t)$ is given in Figure $3(a),$ that varifies Corollary \ref{cor6.3}.
    \begin{figure}[h]
        \begin{center}
            \subfigure[]{\label{c3}\includegraphics[height=3.41in]{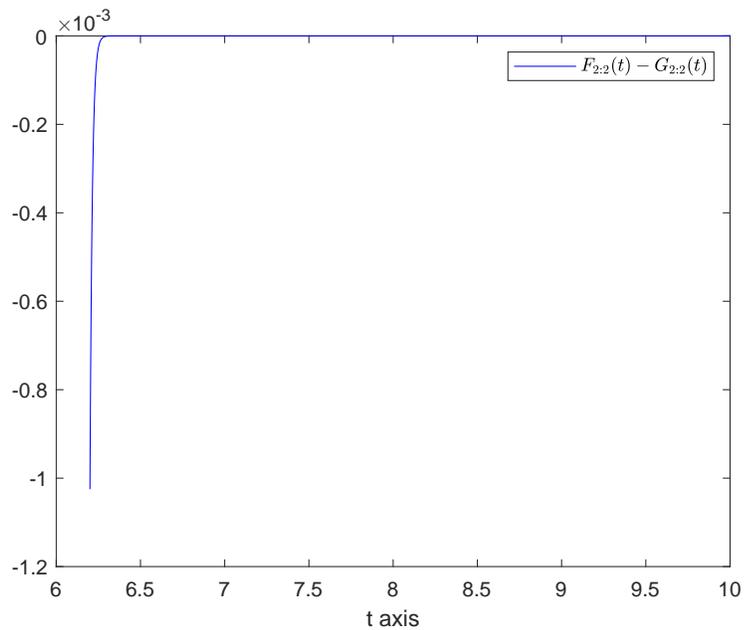}}
            \caption{(a) Graph of $F_{2:2}(t)-G_{2:2}(t)$ for  Example \ref{ex5.2}.}
        \end{center}
    \end{figure}
    \end{example}

\section{Concluding remarks\setcounter{equation}{0}}

Let us have two insurance portfolios of $n$
individual risks. Assume that the portfolios are
heterogeneous. The problem of comparison of the
smallest and largest claim amounts arising from
these portfolios of risks with respect to some
well known stochastic orders is of recent
interest from both theoretical and practical
points of view. Here, under different conditions,
we established stochastic comparisons between the
largest claims in the sense of the usual
stochastic and reversed hazard rate orderings.
Both these orders are useful tools to a decision
maker to choose better one among several risks.
For example, for two risks $X$ and $Y$, if
$X\le_{st}Y,$ then a person will choose $X$ over
$Y$. Again, for the case of the reversed hazard
rate ordering, a person should prefer a bond,
which has smaller reversed hazard rate. The
results have been developed using the concepts of
the vector majorization and related orders, and
the multivariate chain majorization order.
Finally, the established results have been
applied to two baseline distribution functions
for explanation purpose.
\\
\\
{\large\bf Acknowledgements:}  Sangita Das, thanks MHRD, Government of India for financial support. Suchandan Kayal acknowledges the financial support for this research work under a grant numbered MTR/2018/000350, SERB, India.

\bibliography{ref}

@article{dasordering,
  title={ORDERING RESULTS ON EXTREMES OF EXPONENTIATED LOCATION-SCALE MODELS},
  author={Das, Sangita and Kayal, Suchandan and Choudhuri, Debajyoti},
  journal={Probability in the Engineering and Informational Sciences},
  pages={1--24, doi: https://doi.org/10.1017/S0269964819000408 },
  year={2019}
  publisher={Cambridge University Press},
doi={ https://doi.org/10.1017/S0269964819000408}
}

@article{das2019ordering,
  title={Ordering extremes of exponentiated location-scale models with dependent and heterogeneous random samples},
  author={Das, Sangita and Kayal, Suchandan},
  journal={Metrika},
  pages={1--25, doi: https://doi.org/10.1007/s00184-019-00753-2},
  year={2019},
  publisher={Springer},
doi={ https://doi.org/10.1007/s00184-019-00753-2}
}

@article{das2019orderingmo,
  title={Some Ordering Results for the Marshall and Olkin’s Family of Distributions},
  author={Das, Sangita and Kayal, Suchandan},
  journal={Communications in Mathematics and Statistics},
  pages={1--27, doi: https://doi.org/10.1007/s40304-019-00191-6},
  year={2019},
  publisher={Springer},
doi={https://doi.org/10.1007/s40304-019-00191-6}
}

@article{pledger1971comparison,
	title={Comparisons of order statistics and of spacings from  heterogeneous distributions, In Optimizing methods in statistics J.S. Rustagi (ed.), {N}ew {Y}ork: Academic { P}ress},
	author={Pledger, Gordon and Proschan, Frank},
	pages={89--113},
	year={1971},
	publisher={Elsevier}
}

@article{balakrishnan2015stochastic,
  title={Stochastic comparisons of series and parallel systems with generalized exponential components},
  author={Balakrishnan, Narayanaswamy and Haidari, Abedin and Masoumifard, Khaled},
  journal={IEEE Transactions on Reliability},
  volume={64},
  number={1},
  pages={333--348},
  year={2015},
  publisher={IEEE}
}

@article{balakrishnan2018ordering,
  title={Ordering the largest claim amounts and ranges from two sets of heterogeneous portfolios},
  author={Balakrishnan, Narayanaswamy and Zhang, Yiying and Zhao, Peng},
  journal={Scandinavian Actuarial Journal},
  volume={2018},
  number={1},
  pages={23--41},
  year={2018},
  publisher={Taylor \& Francis}
}

@article{barmalzan2015convex,
  title={On the convex transform and right-spread orders of smallest claim amounts},
  author={Barmalzan, Ghobad and Najafabadi, Amir T Payandeh},
  journal={Insurance: Mathematics and Economics},
  volume={64},
  pages={380--384},
  year={2015},
  publisher={Elsevier}
}

@article{barmalzan2016likelihood,
  title={Likelihood ratio and dispersive orders for smallest order statistics and smallest claim amounts from heterogeneous Weibull sample},
  author={Barmalzan, Ghobad and Najafabadi, Amir T Payandeh and Balakrishnan, Narayanaswamy},
  journal={Statistics \& Probability Letters},
  volume={110},
  pages={1--7},
  year={2016},
  publisher={Elsevier}
}

@article{barmalzan2017ordering,
  title={Ordering properties of the smallest and largest claim amounts in a general scale model},
  author={Barmalzan, Ghobad and Payandeh Najafabadi, Amir T and Balakrishnan, Narayanaswamy},
  journal={Scandinavian Actuarial Journal},
  volume={2017},
  number={2},
  pages={105--124},
  year={2017},
  publisher={Taylor \& Francis}
}

@article{hazra2017stochastic,
  title={On stochastic comparisons of maximum order statistics from the location-scale family of distributions},
  author={Hazra, Nil Kamal and Kuiti, Mithu Rani and Finkelstein, Maxim and Nanda, Asok K},
  journal={Journal of Multivariate Analysis},
  volume={160},
  pages={31--41},
  year={2017},
  publisher={Elsevier}
}

@article{khaledi2002dispersive,
  title={Dispersive ordering among linear combinations of uniform random variables},
  author={Khaledi, Baha Eldin and Kochar, Subhash C},
  journal={Journal of Statistical Planning and Inference},
  volume={100},
  number={1},
  pages={13--21},
  year={2002},
  publisher={Elsevier}
}

@article{kundu2016some,
  title={Some results on majorization and their applications},
  author={Kundu, Amarjit and Chowdhury, Shovan and Nanda, Asok K and Hazra, Nil Kamal},
  journal={Journal of Computational and Applied Mathematics},
  volume={301},
  pages={161--177},
  year={2016},
  publisher={Elsevier}
}

@book{Marshall2011,
    author= "A. W. Marshall and I. Olkin and B. C. Arnold",
    title="Inequality: Theory of Majorization and its {A}pplications",
    publisher = "Springer Series in Statistics, New York ",
    year = "2011"
}

@book{shaked2007stochastic,
  title={Stochastic orders},
  author={Shaked, Moshe and Shanthikumar, J George},
  year={2007},
  publisher={Springer Science \& Business Media}
}

@article{zhang2019ordering,
  title={Ordering Properties Of Extreme Claim Amounts From Heterogeneous Portfolios},
  author={Zhang, Yiying and Cai, Xiong and Zhao, Peng},
  journal={ASTIN Bulletin: The Journal of the IAA},
  volume={49},
  number={2},
  pages={525--554},
  year={2019},
  publisher={Cambridge University Press}
}

@article{zardasht2015test,
  title={A test for the increasing convex order based on the cumulative residual entropy},
  author={Zardasht, V},
  journal={Journal of the Korean Statistical Society},
  volume={44},
  number={4},
  pages={491--497},
  year={2015},
  publisher={Elsevier}
}

\end{document}